\documentclass[a4paper,11pt]{article}

\usepackage[margin=1in]{geometry}

\usepackage{hyperref}
\hypersetup{
	colorlinks=true,
	citecolor = blue!80!black,
	linkcolor = red!80!black
}

\newcommand{\email}[1]{\href{mailto:#1}{\nolinkurl{#1}}}
\usepackage[numbers]{natbib}
\renewcommand{\citet}[1]{\citeauthor{#1}~\citep{#1}}

\usepackage[title]{appendix}
\usepackage{amsmath,amsthm,amsfonts}
\usepackage{microtype}
\usepackage{ifthen}
\usepackage{xcolor}
\usepackage{tikz}
\usepackage{pgfplots}
\usetikzlibrary{calc,intersections}
\usepackage{xspace}
\usepackage[footnotesize]{caption}
\usepackage[ruled,vlined]{algorithm2e}

\theoremstyle{plain}
\newtheorem{theorem}{Theorem}
\newtheorem{lemma}{Lemma}
\newtheorem{proposition}{Proposition}
\newtheorem{corollary}{Corollary}
\theoremstyle{definition}
\newtheorem{definition}{Definition}

\newcommand{\ie}{i.e.,\xspace}
\newcommand{\eg}{e.g.,\xspace}

\def\clap#1{\hbox to 0pt{\hss#1\hss}}

\renewcommand{\vec}[1]{\mathbf{#1}}

\newcommand{\vr}{\vec{r}}
\newcommand{\vs}{\vec{s}}

\newcommand{\vx}{\vec{x}}
\newcommand{\R}{\mathbb{R}}
\newcommand{\N}{\mathbb{N}}
\newcommand{\E}[2][]{\mathbb{E}_{#1}\left[#2\right]}
\newcommand{\EE}{\mathbb{E}}
\renewcommand{\Pr}[1]{\mathrm{Pr}\left[#1\right]}
\newcommand{\PP}{\mathrm{Pr}}

\newcommand{\perm}{\mathcal{S}}

\begin{document}

\title{Prophet Inequalities for I.I.D.\@ Random Variables\\ from an Unknown Distribution}

\author{%
	Jos\'e R.\ Correa\thanks{Departamento de Ingenier\'ia Industrial, Universidad de Chile, Santiago, Chile. Email: \email{correa@uchile.cl}}
	\and 
	Paul D\"utting\thanks{Department of Mathematics, London School of Economics, London, UK. Email: \email{p.d.duetting@lse.ac.uk}.}
	\and
	Felix Fischer\thanks{School of Mathematical Sciences, Queen Mary University of London, London, UK. Email: \email{felix.fischer@qmul.ac.uk}.}
	\and
	Kevin Schewior\thanks{Department of Mathematics and Computer Science, University of Cologne, Cologne, Germany. Email: \email{kschewior@gmail.com}.}}

\date{}

\maketitle

\begin{abstract}
	A central object of study in optimal stopping theory is the single-choice prophet inequality for independent, identically distributed random variables: given a sequence of random variables $X_1,\dots,X_n$ drawn independently from the same distribution, the goal is to choose a stopping time~$\tau$ such that for the maximum value of~$\alpha$ and for all distributions, $\mathbb{E}[X_\tau] \geq \alpha \cdot \mathbb{E}[\max_t X_t]$. What makes this problem challenging is that the decision whether $\tau=t$ may only depend on the values of the random variables $X_1,\dots,X_t$ and on the distribution~$F$. For a long time the best known bound for the problem had been $\alpha \geq 1-1/e \approx 0.632$, but recently a tight bound of $\alpha \approx 0.745$ was obtained.
	The case where $F$ is unknown, such that the decision whether $\tau=t$ may depend only on the values of the random variables $X_1,\dots,X_t$, is equally well motivated but has received much less attention. A straightforward guarantee for this case of $\alpha\geq 1/e\approx 0.368$ can be derived from the well-known optimal solution to the secretary problem, where an arbitrary set of values arrive in random order and the goal is to maximize the probability of selecting the largest value. We show that this bound is in fact tight. We then investigate the case where the stopping time may additionally depend on a limited number of samples from~$F$, and show that even with $o(n)$ samples $\alpha\leq 1/e$. On the other hand, $n$ samples allow for a significant improvement, while $O(n^2)$ samples are equivalent to knowledge of the distribution: specifically, with $n$ samples $\alpha\geq 1-1/e\approx 0.632$ and $\alpha\leq\ln(2)\approx 0.693$, and with $O(n^2)$ samples $\alpha\geq 0.745-\varepsilon$ for any $\varepsilon>0$.
\end{abstract}

\maketitle

\section{Introduction}

The theory of optimal stopping is concerned with sequential decision making given imperfect information about the future, in order to maximize a reward or minimize a cost. Two canonical problems in the field are the \emph{secretary problem} and the \emph{prophet problem}. Both problems have over the past few years also received considerable attention from theoretical computer science and operations research, at least in part due to their relevance to the design of posted-price mechanisms for online sales.

In the \emph{secretary problem} we are given $n$ distinct, non-negative numbers from an unknown range. These numbers are presented in random order, and the goal is to stop at one of these numbers in order to maximize the probability with which we select the maximum. The problem has a surprisingly simple, and surprisingly positive, answer: by discarding a $1/e$ fraction of the numbers, and then selecting the first number that is greater than any of the discarded numbers, one is guaranteed to select the maximum with probability~$1/e$~\citep[\eg][]{GiMo66a}. The guarantee of $1/e$ achieved by this simple stopping rule is best possible, and remains best possible for example when numbers come from a uniform distribution with unknown and randomly chosen endpoints and are therefore correlated random variables~\citep{BeGn84a,Ferg89a}. When numbers are i.i.d.\@ from a known distribution, a better guarantee of around~$0.58$ can be obtained~\citep{GiMo66a}, and this bound is again tight.

In the \emph{prophet problem} we are again shown~$n$ non-negative numbers, one at a time, but now these numbers are independent draws from known distributions and our goal is to maximize the expected value of the number on which we stop relative to the expected maximum value in the sequence. Two central results for this problem concern the case where the distributions are distinct and the case where they are identical. For the former a tight bound of~$1/2$ was given by \citet{KrengelS1973,KrengelS1978} and~\citet{SamuelC84}. For the latter a lower bound of $1-1/e\approx 0.632$ due to \citet{HillK82}, corresponding to a stopping rule with this guarantee, was improved only very recently, first to around~$0.738$ by \citet{AbolhassaniEEHK17} and then to around~$0.745$ by \citet{CorreaFHOV17}. The lower bound of \citet{CorreaFHOV17} is in fact known to be tight due to an impossibility result of \citet{HillK82} and \citet{Kertz86} that implies a matching upper bound.

A natural variant of the prophet problem, for both identical and non-identical distributions, can be obtained if we assume that the distributions from which values are drawn are unknown. Despite its obvious appeal, which was noted for example by \citet{AzarKW14}, precious little is known about this variant.

\paragraph{Our Contribution}
We consider the prophet problem in which values are drawn independently from a \emph{single} \emph{unknown} distribution, and ask which approximation guarantees can be obtained relative to the expected maximum value in hindsight. It is worth pointing out that, in contrast to the case where the distribution is known and an optimal stopping rule can be obtained via backward induction, there is no clear candidate for an optimal stopping rule. The case of identical distributions seems particularly interesting, as here one may hope to be able to learn something about later values from earlier ones.

A guarantee of $1/e$ for the problem can be obtained in a relatively straightforward way from the well-known optimal stopping rule for the secretary problem, see \autoref{thm:secretary} in \autoref{sec:sublinear}. The rule is guaranteed to stop on the maximum value with probability at least $1/e$, and one can show that this implies a $1/e$-approximation relative to the expected maximum. Note that such an analysis, however, does not take into account that all values come from the same distribution and thus ignores any possibility of the aforementioned learning.

We show that no learning of the distribution is possible and that the straightforward guarantee of~$1/e$ is in fact best possible in the prophet setting, see \autoref{thm:upper} in \autoref{sec:sublinear}. The main difficulty in showing an impossibility result of this kind is that the set of stopping rules to which it applies is very rich. We will see, however, that for every stopping rule there exists a set $V\subseteq\N$ of arbitrary size and with an arbitrary gap between the largest and second-largest element on which the stopping rule is what we call value-oblivious: for random variables $X_1,\dots,X_n$ supported on~$V$, the decision to stop at~$X_i$ when $X_i>\max\{X_1,\dots,X_{i-1}\}$ does not depend on the values of the random variables $X_1,\dots,X_i$ but only on whether $X_i$ is the largest among these values. We will then construct a distribution~$F$ with support~$V$ such that~$n$ values drawn independently from~$F$ are pairwise distinct with probability arbitrarily close to one and the expectation of their maximum is dominated by the largest value in~$V$. The objective of the prophet problem on~$F$ is thus identical, up to a small error, to that of the secretary problem, and any stopping rule with a guarantee better than~$1/e$ for the former would yield such a stopping rule for the latter. To understand why stopping rules must be value-oblivious it is useful to consider the special case where $n=2$. In this case we may focus on rules that always stop at $X_2$ whenever they have not stopped at $X_1$, and every such stopping rule can be described by a function $p:\R\to[0,1]$ such that $p(x)$ is the probability of stopping at~$X_1$ when $X_1=x$. By the Bolzano--Weierstrass theorem the infinite sequence $(p(n))_{n\in\N}$ contains a monotone subsequence and thus, for some $q\in[0,1]$ and every $\varepsilon>0$, a subsequence of values contained in the interval $[q-\varepsilon,q+\varepsilon]$. For random variables that only take values in the index set of that latter subsequence, the stopping rule will therefore stop at the first random variable with what is essentially a fixed probability. When $n>2$ the set of possible stopping rules becomes much richer, and identifying a set~$V$ on which a particular stopping rule is value-oblivious becomes much more challenging. Rather than the Bolzano--Weierstrass theorem, our proof uses the infinite version of Ramsey's theorem~\citep{Rams30a} to establish the existence of such a set.

\begin{figure}
\centering\hspace*{-2em}
\begin{tikzpicture}[yscale=1.2]
	\draw[color=black!10,fill=black!10] (2.5,-0.5) rectangle (3.5,3.9);	
	\draw[-latex,very thick] (0,-0.5) -- (8.5,-0.5);
	\draw[-,very thick] (0,-0.5) -- (0,-0.2);
	\draw[-,very thick] (-0.2,-0.1) -- (0.2,-0.1);
	\node[below] at (1.25,-0.55) {\footnotesize $o(n)$};
	\node[below left] at (2.8,-0.64) {\footnotesize $\smash\varepsilon\, n$};
	\node[below right] at (3.3,-0.64) {\footnotesize $n$};
	\node[below] at (5.35,-0.56) {\footnotesize $\Omega(n)$};
	\node[below] at (7.75,-0.51) {\footnotesize $\Theta(n^2)$};
	\draw[-latex,very thick] (0,0) -- (0,4);
	\draw[-,very thick] (-0.05,0.45) node[left] {\footnotesize $1/e$} -- (0.05,0.45) ;
	\draw[-,very thick,dashed] (0,0.47) to node[above,color=black]{\footnotesize \autoref{thm:upper}} (2.5,0.47) to (3.1,2) to [bend left=5] (3.8,3.35) to node[above,color=black]{\footnotesize \citet{HillK82,Kertz86}} (7.7,3.35);
	\draw[-,very thick] (0,0.43) -- node[below,color=black]{\footnotesize \autoref{thm:secretary}} (2.665,0.43) -- (3.5,2.5);
	\draw[-,very thick] (-0.05,2.9) node[left] {\footnotesize $\ln(2)$} -- (0.05,2.9); 
	\node[left,color=black] at (3.3,2.9) {\footnotesize \autoref{thm:upper_par}};
	\draw[-,very thick] (-0.05,2.5) node[left] {\footnotesize $1-1/e$} -- (0.05,2.5); 
	\node[circle,fill=black,minimum size=3pt,inner sep=0pt] at (3.5,2.5) {};
	\node[right,color=black] at (3.5,2.7) {\footnotesize \autoref{thm:fresh-looking-samples}};
	\node[right,color=black] at (3.1,1.5) {\footnotesize \autoref{cor:parametric-lower}};	
	\draw[-,very thick] (3.5,2.5) -- (3.7,2.5);
	\draw[very thick,dotted] (3.6,2.5) -- (4,2.5);	
	\draw[-,very thick] (-0.05,3.35) node[left] {\footnotesize $0.745$} -- (0.05,3.35); 
	\node[circle,fill=black,minimum size=3pt,inner sep=0pt] at (7.75,3.35) {};
	\node[below] at (7.75,3.35) {\footnotesize \autoref{thm:quadratic}};
\end{tikzpicture}
\caption{Overview of results. The number of samples is displayed along the horizontal axis, the performance guarantee along the vertical axis. Lower bounds, shown as a solid line and two dots, result from stopping rules with a certain performance guarantee. Upper bounds, shown as dashed lines, correspond to impossibility results that no stopping rule can improve upon. The results for $o(n)$ and $\Theta(n^2)$ samples are tight. With the exception of the upper bound of approximately $0.745$, all results are new to this paper.}
\end{figure}

Motivated by this impossibility result we then turn to the case where the stopping rule has access to a limited number of additional samples from the distribution, which it may use in deciding when to stop. An extension of our impossibility result shows that $o(n)$ samples are not enough to improve on the bound of~$1/e$. The interesting case therefore is the one with $\Omega(n)$ samples, and we show that a simple stopping rule achieves a guarantee of $1-1/e\approx 0.632$ with $n-1$ samples, see \autoref{thm:fresh-looking-samples} in \autoref{sec:linear}. 
The rule starts by drawing $n-1$ samples. Then, when considering the $i$th random variable for $i \geq 1$, it also considers a random subset of size $n-1$ drawn uniformly from the $n-1$ initial samples and the $i-1$ random variables seen so far. If the $i$th random variable is greater than the maximum of that random subset the rule stops, otherwise it continues with the next random variable. While the stopping rule itself is easy to describe, its analysis relies on an insight that is somewhat subtle. Indeed, each of the sets of random variables used to set a threshold for acceptance is distributed like a set of $n-1$ fresh samples from the distribution. The expected value collected from each random variable, conditioned on its acceptance, thus equals the expected maximum value of $n$ independent draws from the distribution, and the probability of accepting a random variable conditioned on reaching it is exactly $1/n$. The approximation guarantee is then equal to the overall probability of stopping, which is at least $1-1/e$. By a straightforward extension, \autoref{cor:parametric-lower} in \autoref{sec:linear}, we obtain a lower bound of $\frac{1+\gamma}{2}(1-1/e)$ on the guarantee achievable with $\gamma\, n$ samples for any $\gamma\in[0,1]$.

We complement the lower bound of $1-1/e$ with matching upper bounds for two different classes of stopping rules that share specific properties of the stopping rule described above. These bounds limit the types of approaches that could conceivably be used to go beyond a performance guarantee of $1-1/e$. We then give a parametric upper bound that applies to any stopping rule with access to $\gamma\, n$ samples for $\gamma\geq 0$, see \autoref{thm:upper_par} in \autoref{sec:linear}. For rules that use at most $n$ samples this upper bound is equal to $\ln(2)\approx 0.693$ and thus nearly tight.

We finally show that $O(n^2)$ samples are enough to get arbitrarily close to the optimal guarantee of around $0.745$ attainable when the distribution is known, see \autoref{thm:quadratic} in \autoref{sec:745}. This is achieved by mimicking the stopping rule that attains that bound, which uses a decreasing sequence of thresholds corresponding to conditional acceptance probabilities that increase over time, but using quantiles of the empirical distribution rather than the actual one. By discarding a constant initial fraction of the values, and using the inequality of \citet*{DKW} to show simultaneous concentration of all empirical quantiles, we reduce the number of required samples from $O(n^4)$ to $O(n^2)$ relative to the obvious approach that potentially stops on any of the values and uses Chernoff and union bounds to show concentration.

Taken together our results reveal a phase transition from secretary-like to prophet-like behavior when going from $o(n)$ to $\Omega(n)$ samples, and show that $O(n^2)$ samples are equivalent to full knowledge of the distribution.

\paragraph{Follow-Up Work} 
\citet{RubinsteinWW20} subsequently showed that a guarantee arbitrarily close to~$0.745$ can in fact already be achieved with $O(n)$ samples. \citet{KaplanNR20} and \citet{CorreaCES20} studied generalizations of the problem we consider here. The results of \citeauthor{KaplanNR20} imply a lower bound for our problem with $\gamma n$ samples of $e^{-e^{-\gamma}}$ when $0\leq\gamma\leq 0.567$ and of around $\gamma (1-\gamma-e^{-\gamma})$ when $0.567\leq\gamma\leq 1$. This improves on our lower bound when $\gamma<1$ and matches our lower bound when $\gamma=1$. The results of \citeauthor{CorreaCES20}, on the other hand, imply an improved lower bound of $0.635$ when $\gamma=1$, \ie for a situation with $n$ samples.

\paragraph{Further Related Work}
For early work on the classic single-choice prophet inequality in mathematics the reader is referred to a survey of \citet{HillK92}. Starting from work of \citet{HajiaghayiKS07}, the prophet problem and extensions to richer feasibility conditions have seen a surge of interest in theoretical computer science. This has produced prophet inequalities for matroids and polymatroids~\citep[\eg][]{ChawlaHMS10,KlWe12,Alaei14,DuettingK15,FeldmanSZ16,RubinsteinS17,AnariNSS19}, settings where feasible solutions are given by an arbitrary downward-closed set system~\citep[\eg][]{Rubinstein16,RubinsteinS17}, matching problems~\citep[\eg][]{ChawlaHMS10,KlWe12,GravinW19,EzraFGT20}, knapsack constraints~\citep[\eg][]{FeldmanSZ16,DuettingFKL17}, resource allocation problems involving intervals and paths~\citep[\eg][]{CMT19a}, and combinatorial auctions~\citep[\eg][]{FeldmanGL15,DuettingFKL17,DuettingKL20}.

There also exists a relatively small but important body of prior work on the case of unknown distributions. Most relevant for us is the aforementioned work by \citet{AzarKW14}, which focuses on richer feasibility structures such as matchings and matroids, and work by \citet{BabaioffBDS17}, who consider a setting similar to ours but focus on a different objective, revenue maximization, apply different techniques, and obtain results that are qualitatively different from ours. Independently from our work, \citet{RubinsteinWW20} studied the case of random variables that are drawn independently from unknown \emph{non-identical} distributions, and showed that a single sample from each distribution is enough to achieve a guarantee of $1/2$, which matches the best possible guarantee that can be achieved when the distributions are known.

In another variant of the prophet problem, the so-called prophet secretary problem, random variables are drawn from known non-identical distributions and observed in random order. The goal is again to immediately and irrevocably choose random variables with large value. \citet{EHLM15a} gave a lower bound of $1-1/e$ for the version of the problem where a single random variable must be chosen, which was subsequently improved to $1-1/e+1/400$ by \citet{ACK18a}. The best bounds currently known for the single-choice version are a lower bound of $1-1/e+1/27\approx 0.669$ and an upper bound of $\sqrt{3}-1\approx 0.732$ due to \citet{CSZ19a}. Combinatorial versions of the prophet secretary problem have been studied by \citet{EhsaniHKS18}. The aforementioned bounds of \citet{KaplanNR20} and \citet{CorreaCES20} also apply to a single-choice prophet secretary problem with unknown distributions and additional samples.

Optimal stopping with unknown distributions has also been studied in operations research and management science, but the types of problems, objectives, and techniques differ significantly from ours and typically involve regret minimization, see, \eg the results of \citet{GoZe17a} and the recent survey of \citet{Boer15}. The literature in operations research and management science moreover contains results on a broad range of stochastic optimization problems, which share certain features of the basic prophet inequality problem and its combinatorial extensions. These problems include constrained Bayesian online selection~\citep[\eg][]{AnariNSS19}, Bayesian assortment optimization~\citep[\eg][]{FengEtAl19}, and certain models of network revenue management~\citep[\eg][]{alijani2020predict}. To the best of our knowledge, they have not been studied from the perspective of sample complexity.

A final line of related work, in economics and theoretical computer science, has studied posted pricing and prophet inequalities with inaccurate priors~\citep{BergemannS08,BergemannS11,DuettingK19}. This line of work assumes access to prior distributions that are close in terms of some metric to the actual distributions, and seeks either max-min optimal mechanisms or performance guarantees that are parametrized by the distance between the assumed and actual priors.

\section{Preliminaries}

Denote by $\N$ the set of positive integers and let $\N_0=\N\cup\{0\}$. For $i\in\N$, let $[i]=\{1,\dots,i\}$ and denote by~$\perm_i$ the set of permutations of~$[i]$.

Let $k\in\N_0$ and $n\in\N$. We consider $(k,n)$-stopping rules that sequentially observe random variables $X_1, \dots, X_n$ and have access to samples $S_1, \dots, S_k$, and for each $i = 1, \dots, n$ decide whether to stop on $X_i$ based on the values of $X_1, \dots, X_i$ and $S_1, \dots, S_k$. We assume that $X_1, \dots, X_n$ and $S_1, \dots, S_k$ are independent and identically distributed, and respectively denote by~$f$ and~$F$ the probability density function and cumulative distribution function of their distribution. 
Formally, a \emph{$(k,n)$-stopping rule} $\vr$ is a family of functions $r_1,\dots,r_n$ where $r_i:\R_+^{k+i}\rightarrow[0,1]$ for all $i=1,\dots,n$. Here, $r_i(s_1\dots,s_k,x_1,\dots,x_i)$ for $\vs\in\R_+^k$ and $\vx\in\R_+^n$ is the probability of stopping at $X_i$ conditioned on having received $S_1=s_1\dots,S=s_k$ as samples and $X_1=x_1,\dots,X_i=x_i$ as values and not having stopped on any of $X_1,\dots,X_{i-1}$. The \emph{stopping time} $\tau$ of a $(k,n)$-stopping rule $\vr$, given $S_1,\dots,S_k$ and $X_1,\dots,X_n$, is thus the random variable with support $\{1,\dots,n\}\cup\{\infty\}$ such that for all $\vs\in\R_+^k$ and $\vx\in\R_+^n$,
\begin{multline*}
	\Pr{\tau=i\mid S_1=s_1,\dots,S_k=s_k,X_1=x_1,\dots,X_n=x_n} = \\ \Biggl(\prod_{j=1}^{i-1} \bigl(1-r_j(s_1,\dots,s_k,x_1,\dots,x_j)\bigr)\Biggr)
	\cdot r_i(s_1\dots,s_k,x_1,\dots,x_i) .
\end{multline*}

For a given stopping rule we will be interested in the expected value $\E{X_\tau}$ of the variable at which it stops, where we use the convention that $X_{\infty}=0$, and will measure its performance relative to the expected maximum $\E{\max\{X_1,\dots,X_n\}}$ of the random variables $X_1,\dots,X_n$. We will say that a stopping rule achieves an approximation guarantee~$\alpha$, for $\alpha\leq 1$, if for any distribution, $\E{X_\tau}\geq\alpha\,\E{\max\{X_1,\dots,X_n\}}$.

For ease of exposition we will assume continuity of~$F$ in proving lower bounds and mainly use discrete distributions to prove upper bounds. All results can be shown to hold in general by standard arguments, to break ties among random variables and to approximate a discrete distribution by a continuous one.

\section{Sublinear Number of Samples}
\label{sec:sublinear}

We begin by showing that for $o(n)$ samples, the prophet problem with an unknown distribution behaves like the secretary problem. As we will see in \autoref{sec:secretary}, a straightforward baseline can be obtained from the optimal solution to the secretary problem, which discards a $1/e$ fraction of the values and then accepts the first value that exceeds the maximum of the discarded values. The algorithm does not require any samples, is guaranteed to stop at the maximum of the sequence with probability $1/e$, and can be shown to also provide a $1/e$ approximation for our objective. Our main result in this section, which we prove in \autoref{sec:impossibility}, shows that the bound of $1/e$ is in fact best possible. This results continues to hold with $o(n)$ samples.

\subsection{A $\mathbf{1/e}$-Approximation Without Samples}
\label{sec:secretary}

The following result translates the guarantee of $1/e$ for the secretary problem to a prophet inequality for independent random variables from an unknown distribution.
\begin{theorem}  \label{thm:secretary}
	Let $X_1,X_2,\dots,X_n$ be i.i.d.\@ random variables drawn from an unknown distribution $F$.
	Then there exists a $(0,n)$-stopping rule with stopping time $\tau$ such that
	\begin{align*}
		\E{X_\tau} \geq \frac{1}{e} \cdot \E{\max\{X_1,X_2,\dots,X_n\}} .
	\end{align*}
\end{theorem}

The result can be shown in a straightforward way, based on the idea that the realizations of the random variables $X_1,\dots,X_n$ can be obtained by drawing~$n$ values from their common distribution and then permuting them uniformly at random. The classic analysis of the secretary problem~\citep{Ferg89a} implies that \emph{for each realization of the~$n$ draws}, the optimal stopping rule for this problem obtains the maximum value with probability $1/e$. It thus also obtains at least a $1/e$ fraction of the expected value of this maximum. We formalize this idea and prove \autoref{thm:secretary} in \autoref{app:secretary}.

\subsection{A Matching Upper Bound}
\label{sec:impossibility}

We next show that it is impossible to improve on the straightforward lower bound of~$1/e$.
\begin{theorem}
	\label{thm:upper}
	Let $\delta>0$. Then there exists $n_0\in\N$ such that for any $n\geq n_0$ and any $(0,n)$-stopping rule with stopping time $\tau$ there exists a distribution $F$, not known to the stopping rule, such that when $X_1,\dots,X_n$ are i.i.d.\@ random variables drawn from $F$,
	\begin{align*}
		\EE[X_\tau] \leq \bigg(\frac{1}{e}+\delta\bigg) \cdot \EE[\max\{X_1,\dots,X_n\}]_{}.
	\end{align*}
\end{theorem}

The main difficulty in showing an impossibility result of this kind is that it applies to the set of all possible $(0,n)$-stopping rules, which a priori is very rich. Indeed, recall that a $(0,n)$-stopping rule $\vr$ can be \emph{any} family of functions $r_1,\dots,r_n$ where $r_i:\R_+^{i}\rightarrow[0,1]$ for all $i=1,\dots,n$. 
Our main structural insight will be that we can restrict attention to a much simpler class of stopping rules~$\vr$ that are in a certain sense oblivious to the \emph{values} of the random variables they observe. For random variables $X_1,\dots,X_n$ supported on arbitrarily large sets $V\subseteq\N$, under the condition that $X_1,\dots,X_i$ are pairwise distinct and $X_i>\max\{X_1,\dots,X_{i-1}\}$, and up to an arbitrarily small error~$\varepsilon$, the probability that~$\vr$ stops on $X_i$ will not depend on the values of any of the random variables $X_1,\dots,X_{i}$.
This is made precise by the following definition. Although it is not needed for proving \autoref{thm:upper}, we will consider the more general case of $(k,n)$-stopping rules for any $k\in\N_0$. The structural result extends easily to the more general case, and we will use it later to generalize \autoref{thm:upper}.
\begin{definition}\label{def:val-obl}
Let $\varepsilon>0$, $k\in\N_0$, and $V\subseteq\N$. A $(k,n)$-stopping rule $\vr$ is \emph{$\varepsilon$-value-oblivious on $V$} if, for all $i\in[n]$, there exists $q_i\in[0,1]$ such that, for all pairwise distinct $s_1,\dots,s_k,v_1,\dots,v_i\in V$ with $v_i>\max\{s_1,\dots,s_k,v_1,\dots,v_{i-1}\}$, it holds that $r_i(s_1,\dots,s_k,v_1,\dots,v_i)\in[q_i-\varepsilon,q_i+\varepsilon)$.
\end{definition}

While value-obliviousness significantly restricts the expressiveness of a stopping rule, this restriction turns out to be essentially without loss when it comes to the ability of achieving a certain guarantee across all possible distributions: for any stopping rule and any $\varepsilon>0$, there exists a stopping rule with the same guarantee that is $\varepsilon$-value-oblivious for some infinite set $V\subseteq\N$. This is made precise by the following lemma, which we prove in \autoref{sec:structural}.

\begin{lemma}
	\label{lem:structure1}
	Let $\varepsilon>0$ and $k\in\N_0$. If there exists a $(k,n)$-stopping rule with guarantee $\alpha$, then there exists a $(k,n)$-stopping rule~$\vr$ with guarantee $\alpha$ and an infinite set $V\subseteq\N$ such that $\vr$ is $\varepsilon$-value-oblivious on~$V$.
\end{lemma}

With \autoref{lem:structure1} at hand it is not difficult to prove \autoref{thm:upper}. For any $(0,n)$-stopping rule and an appropriate value of $\varepsilon$, we identify a $(0,n)$-stopping rule~$\vr$ with the same performance guarantee that is $\varepsilon$-value-oblivious on an infinite set $V\subseteq \N$. We then define a distribution~$F$ with finite support~$S \subseteq V$ such that (i)~there is a large gap between the largest and second-largest elements of~$S$, (ii)~$n$ independent draws from~$F$ are pairwise distinct with probability close to~$1$, (iii)~$\vr$ is $\varepsilon$-value-oblivious on $S$, and (iv)~the performance guarantee of~$\vr$ on the distribution is dominated by the probability of selecting the largest element of~$S$. By~(i) and~(ii) the prophet problem for the unknown distribution~$F$ is then equivalent up to a small error to a secretary problem, and by~(iii) and~(iv)~$\vr$ behaves on~$F$ essentially like a stopping rule for the secretary problem. A performance guarantee for~$\vr$ of more than $1/e$ would thus contradict the optimality of this bound for the secretary problem.
\begin{proof}[Proof of \autoref{thm:upper}]
	It suffices to show that the guarantee of any $(0,n)$-stopping rule is bounded from above by $1/e+o(1)$, where implicitly $n\rightarrow\infty$.
	To this end consider an arbitrary $(0,n)$-stopping rule with guarantee $\alpha$. Let $\varepsilon=1/n^2$. 
	By \autoref{lem:structure1} there then exists a $(0,n)$-stopping rule~$\vr$ with guarantee~$\alpha$ and an infinite set $V\subseteq\N$ on which~$\vr$ is $\varepsilon$-value-oblivious. Denote by $\tau$ the stopping time of~$\vr$. Let $v_1,\dots,v_{n^3},u\in V$ be pairwise distinct such that $u\geq n^3\max\{v_1,\dots,v_{n^3}\}$. For each $i\in[n]$, let
\[
	X_i = \begin{cases}
		v_1 & \text{with probability $\frac{1}{n^3}(1-\frac{1}{n^2})$,} \\
		\vdots & \\
		v_{n^3} &\text{with probability $\frac{1}{n^3}(1-\frac{1}{n^2})$,} \\
		u &\text{with probability $\frac{1}{n^2}$.}
	\end{cases}
\]

We proceed to bound $\EE[\max\{X_1,\dots,X_n\}]$ from below and $\EE[X_\tau]$ from above. For $i\in[n]$, let $X_{(i)}$ denote the $i$th order statistic of $X_1,\dots,X_n$, such that $X_{(n)}=\max\{X_1,\dots,X_n\}$. Then
\begin{equation}\label{eq:thm2-opt}
	\EE[\max\{X_1,\dots,X_n\}] \geq \PP[X_{(n)}=u]\cdot u=\frac{1-o(1)}{n}\cdot u.
\end{equation}
On the other hand,
\begin{align}
	\EE[X_\tau] &= \PP[X_{(n)}=u\wedge X_{(n-1)}\neq u]\cdot\EE[X_\tau\mid X_{(n)}=u\wedge X_{(n-1)}\neq u] \nonumber \\
	& \phantom{={}} + \PP[X_{(n)}=u\wedge X_{(n-1)}=u]\cdot\EE[X_\tau\mid X_{(n)}=u\wedge X_{(n-1)}=u] \nonumber \\
	& \phantom{={}} + \PP[X_{(n)}\neq u]\cdot\EE[X_\tau\mid X_{(n)}\neq u] \nonumber \\[4pt]
	& \leq \frac{1}{n} \Bigl(\PP[X_\tau=X_{(n)}\mid X_{(n)}=u\wedge X_{(n-1)}\neq u]\cdot u \nonumber \\[-.5ex]
	& \phantom{=\frac{1}{n}\big({}}+\PP[X_\tau\neq X_{(n)}\mid X_{(n)}=u\wedge X_{(n-1)}\neq u]\cdot O(n^{-3})\cdot u\Bigr) \nonumber \\
	& \phantom{={}} + O(n^{-2})\cdot u+1\cdot O(n^{-3})\cdot u \nonumber \\[1ex]
 	&\leq \frac{1}{n} \PP[X_\tau=X_{(n)}\mid X_{(n)}=u\wedge X_{(n-1)}\neq u]\cdot u+o\Bigl(\frac1n\Bigr)\cdot u \nonumber \\[1ex]
  &\leq\frac{1}{n} \PP[X_\tau=X_{(n)}\mid X_{(n)}=u\wedge \text{$X_1,\dots,X_n$ distinct}]\cdot u + o\Bigl(\frac{1}{n}\Bigr) \cdot u,\label{eq:thm2-alg}
\end{align}
where for the first inequality we have applied the law of total expectation to $\EE[X_\tau\mid X_{(n)}=u\wedge X_{(n-1)}\neq u]$ to additionally condition on whether $X_{\tau}=X_{(n)}$, and for the last inequality we have applied the law of total probability to $\PP[X_\tau=X_{(n)}\mid X_{(n)}=u\wedge X_{(n-1)}\neq u]$ to additionally condition on whether $X_1,\dots,X_n$ are distinct.

Given~\eqref{eq:thm2-opt} and~\eqref{eq:thm2-alg}, to show that $\alpha\leq 1/e+o(1)$ it now suffices to show that
\begin{align}
 \PP[X_\tau=X_{(n)}\mid X_{(n)}=u\wedge \text{$X_1,\dots,X_n$ distinct}]\leq 1/e+o(1). \label{eq:little-oh}
\end{align}

Note that in the event where $X_{(n)}=u$ and $X_1,\dots,X_n$ are pairwise distinct, the relative ranks of $X_1,\dots,X_n$ are distributed uniformly at random. For a $0$-value-oblivious stopping rule $\hat{\vr}$ with stopping time $\hat{\tau}$ it thus follows from the well-known optimal solution to the secretary problem \citep[Section 2]{Ferg89a} that
\begin{equation}
\PP[X_{\hat{\tau}}=X_{(n)}\mid X_{(n)}=u\wedge \text{$X_1,\dots,X_n$ distinct}]\leq 1/e+o(1). \label{eq:little-oh2}
\end{equation}
To show this claim for stopping rule $\vr$, which is only $\varepsilon$-value-oblivious with $\varepsilon > 0$, we construct from $\vr$ a $0$-value-oblivious stopping rule $\hat{\vr}$ and show through a coupling argument that the probability that $\vr$ stops at $X_{(n)}$ is bounded by the probability that $\hat{\vr}$ stops at $X_{(n)}$ plus $n\varepsilon=1/n$.

Since $\vr$ is $\varepsilon$-value-oblivious on $V$, $r_i(s_1,\dots,s_{i})\in[q_i-\varepsilon,q_i+\varepsilon)$ for all $i\in[n]$, some $q_i\in[0,1]$, and all distinct $s_1,\dots,s_{i}\in V$ with $s_i>\max\{s_1,\dots,s_{i-1}\}$. Let $\hat{\vr}$ be the stopping rule such that for all $s_1,\dots,s_{i}\in V$, $\hat{r}_i(s_1,\dots,s_i)=q_i$ if $s_i>\max\{s_1,\dots,s_{i-1}\}$ and $\hat{r}_i(s_1,\dots,s_i)=0$ otherwise. Denote by $\hat{\tau}$ the stopping time of $\hat{\vr}$.

Let $s_1,\dots,s_n\in V$ be distinct and assume that $X_1=s_1,\dots,X_n=s_n$. To compare the performance of $\vr$ and $\hat{\vr}$, we can view $\tau$ and $\hat{\tau}$ as being coupled via $n$ independent draws $c_1,\dots,c_n$ from the uniform distribution on $[0,1]$. For every $i\in[n]$, and under the condition that $\tau\geq i$, we can assume that $\tau=i$ if and only if $s_i>\max\{s_1,\dots,s_{i-1}\}$ and $r_i(s_1,\dots,s_i)>c_i$. Similarly, for every $i\in[n]$, and under the condition that $\hat{\tau}\geq i$, we can assume that $\hat{\tau}=i$ if and only if $s_i>\max\{s_1,\dots,s_{i-1}\}$ and $\hat{r}_i(s_1,\dots,s_i)=q_i>c_i$. For $i\in[n]$, let $\xi_i$ be the event that occurs if and only if $s_i>\max\{s_1,\dots,s_{i-1}\}$ and $c_i\in [\min\{r_i(s_1,\dots,s_i),\hat{r}_i(s_1,\dots,s_i)\},\max\{r_i(s_1,\dots,s_i),\hat{r}_i(s_1,\dots,s_i)\}]$. Then $\PP[\xi_i]\leq\varepsilon$, while $X_\tau=X_{(n)}\neq X_{\hat{\tau}}$ requires $\xi_i$ to occur for some $i\in[n]$. Thus, by the union bound,
\[\PP[X_{\tau}=X_{(n)}\mid X_1=s_1,\dots,X_n=s_n]\leq\PP[X_{\hat{\tau}}=X_{(n)}\mid X_1=s_1,\dots,X_n=s_n]+n\varepsilon.\]
Since this statement holds pointwise for all distinct $s_1,\dots,s_n\in V$,
\begin{multline}
	\PP[ X_{\tau}=X_{(n)}\mid X_{(n)}=u\wedge \text{$X_1,\dots,X_n$ distinct}] \leq
	 \\[1ex] \PP[X_{\hat{\tau}}=X_{(n)}\mid X_{(n)}=u\wedge \text{$X_1,\dots,X_n$ distinct}]+n\varepsilon.\label{eq:thm2-unionbound}
\end{multline}
Substituting~\eqref{eq:thm2-unionbound} into~\eqref{eq:little-oh2} yields~\eqref{eq:little-oh}, which completes the proof.
\end{proof}

\subsection{Proof of \autoref{lem:structure1}}  \label{sec:structural}

The lemma claims that for every $\varepsilon > 0$, the existence of a $(k,n)-$stopping rule with performance guarantee $\alpha$ implies that of a $(k,n)$-stopping rule~$\vr$ with the same performance guarantee and of an infinite set $V \subseteq \N$ on which~$\vr$ is $\varepsilon$-value oblivious. Since we can interpret a $(k,n)$-stopping rule as a $(0,n')$-stopping rule with $n' = k+n$ that never stops on the first $k$ values, it will be sufficient to consider $(0,n)$-stopping rules with this additional constraint.
We prove the lemma through a sequence of steps that successively restrict the expressiveness of the stopping rules we have to consider. First we show a restriction to what we call order-oblivious rules, which in the decision to stop at random variable $X_i$, and conditioned on having reached~$X_i$, may take into account the values of random variables $X_1,\dots,X_{i-1}$ but not the order in which they were observed.
\begin{definition}  \label{def:order}  %
	A $(0,n)$-stopping rule $\vr$ is \emph{order-oblivious} if for all $j\in[n]$, all pairwise distinct $v_1,\dots,v_j\in\R_+$ and all permutations $\pi\in \perm_{j-1}$, $r_i(v_1,\dots,v_j)=r_i(v_{\pi(1)},\dots,v_{\pi(j-1)},v_{j})$.
\end{definition}
The following result is rather intuitive.
\begin{lemma}  \label{lem:structure2}
	If there exists a $(0,n)$-stopping rule~$\vr$ with guarantee $\alpha$, then there exists a $(0,n)$-stopping rule~$\vr'$ with guarantee $\alpha$ that is order-oblivious and that, for any $i\in[n]$, never selects $X_i$ if $\vr$ never selects $X_i$.
\end{lemma}

A naive attempt to prove this lemma would be to construct an order-oblivious stopping rule from an arbitrary stopping rule $\vr$ by permuting $X_1,\dots,X_{i-1}$ uniformly at random upon observing~$X_i$, and accepting $X_{i}$ if and only if $\vr$ would accept it under the random permutation. The resulting stopping rule may, however, have a different guarantee than $\vr$ because the probability that $\vr$ arrives at $X_i$ may vary depending on the permutation. Some additional care is therefore required.
\begin{proof}[Proof of \autoref{lem:structure2}]
	For $i\in[n]$, let $\sim_i$ be the equivalence relation on $\R_+^i$ such that $(v_1,\dots,v_i)\sim_i(w_1,\dots,w_i)$ if $v_1,\dots,v_{i-1}$ is a permutation of $w_1,\dots,w_{i-1}$ and $v_i=w_i$. Note that a stopping rule $\vr$ with stopping time $\tau$ is order-oblivious if and only if for all $i\in[n]$ and $v_1,\dots,v_{i},w_1,\dots,w_{i}\in\R_+$ it holds that $r_i(v_1,\dots,v_{i})=r_i(w_1,\dots,w_{i})$ whenever $(v_1,\dots,v_i)\sim_i(w_1,\dots,w_i)$. We will refer to the equivalence classes of $\sim_i$ as \emph{states}, and will say that $\vr$ \emph{arrives at $s\in\R_+^i/\sim_i$} in the event that $\tau\geq i$ and $X_1=v_1,\dots,X_{i-1}=v_{i-1}$ where $[v_1,\dots,v_i]_{\sim_i}=s$.
	
	Let $\vr$ be an arbitrary stopping rule with stopping time $\tau$, and define a stopping rule $\vr'$ with stopping time $\tau'$ such that $r_1(v_1)=r'_1(v_1)$ and for all $i\in\{2,\dots,n\}$ and $v_1,\dots,v_i\in\R_+$ with $\Pr{\text{$\vr$ arrives at $[v_1,\dots,v_i]_{\sim_i}$}}>0$,
\begin{align*}
	r'_i(v_1,\dots,v_i) = &\;\PP\bigl[\tau=i \mid \text{$\vr$ arrives at $[v_1,\dots,v_i]_{\sim_i}$}\bigr].
\end{align*}
Since $[v_1,\dots,v_i]_{\sim_i}$ is invariant under permutations of the sequence $v_1,\dots,v_{i-1}$, $\vr'$ is indeed order-oblivious. Moreover, for any $i\in[n]$, if $\vr$ never selects $X_i$ then neither does $\vr'$. It remains to be shown that $\vr'$ provides guarantee $\alpha$.

As an intermediate step we show by induction that for all $i\in[n]$ and $s\in\R_+^i/\sim_i$,
\begin{align}  \label{eq:state-distr}
	\PP[\text{$\vr$ arrives at $s$}]=\PP[\text{$\vr'$ arrives at $s$}].
\end{align}
This holds trivially for $i=1$, so we assume that it holds for $i=k-1\geq 1$ and show then that it holds for $i=k$. For any $v_1,\dots,v_k\in\R_+$ and $s=[v_1,\dots,v_k]_{\sim_k}$, 
we write $v_{-j} = (v_1, \dots, v_{j-1},v_{j+1}, \dots, v_{k-1})$ for the sequence of length $k-2$ in which $v_j$ has been left out and $(v_{-j},v_j)=(v_1,\dots,v_{j-1},v_{j+1},\dots,v_{k-1},v_j)$ for the sequence of length $(k-1)$ obtained by appending $v_j$ to $v_{-j}$. Then
\begin{align*}
	\PP[&\text{$\vr$ arrives at $s$}] \\
	&= \sum_{j=1}^{k-1}\PP\big[\vr\text{ arrives at $[v_{-j},v_j]_{\sim_{k-1}}$}\big] \cdot \PP\big[\tau\neq i \mid \text{$\vr$ arrives at $[v_{-j},v_j]_{\sim_{k-1}}$}\big] \cdot\Pr{X_k=v_k} \\[.5ex]
	&= \sum_{j=1}^{k-1}\PP\big[\text{$\vr'$ arrives at $[v_{-j},v_j]_{\sim_{k-1}}$}\big] \cdot \PP\big[\tau'\neq i \mid \text{$\vr'$ arrives at $[v_{-j},v_j]_{\sim_{k-1}}$}\big] \cdot \Pr{X_k=v_k} \\[.5ex]
	&= \PP[\text{$\vr'$ arrives at $s$}],
\end{align*}
where the first and last equalities hold by definition of $\sim_{k-1}$ and the second equality by the induction hypothesis and by definition of $\vr'$.

We now claim that
\begingroup\allowdisplaybreaks
\begin{align*}
	\E{X_\tau} &= \sum_{i=1}^n\E{X_i \mid \tau=i} \cdot \Pr{\tau=i} \\
	&= \sum_{i=1}^n\int_0^\infty\dots\int_0^\infty \prod_{j=1}^i f(v_j) \cdot v_i \\
	&\hspace{1cm} \cdot \frac{1}{(i-1)!} \cdot \sum_{\pi\in \perm_{i-1}} \Pr{\tau=i\mid X_1=v_{\pi(1)},\dots,X_{i-1}=v_{\pi(i-1)},X_i=v_i}\;\mathrm{d}v_1\dots\mathrm{d}v_i \\
	&= \sum_{i=1}^n\int_0^\infty\dots\int_0^\infty \prod_{j=1}^i f(v_j) \cdot v_i \cdot \PP\big[\tau=i \mid \text{$\vr$ arrives at $[v_1,\dots,v_i]_{\sim_i}$}\big] \\[-3ex]
	& \phantom{{}=\sum_{i=1}^n\int_0^\infty\dots\int_0^\infty \prod_{j=1}^i f(v_j) \cdot v_i\cdot r'_i(v_1,\dots,v_i)}  \cdot \PP\big[\text{$\vr$ arrives at $[v_1,\dots,v_i]_{\sim_i}$}\big]\;\mathrm{d}v_1\dots\mathrm{d}v_i \\[-1.5ex]
	&= \sum_{i=1}^n\int_0^\infty\dots\int_0^\infty \prod_{j=1}^i f(v_j) \cdot v_i\cdot r'_i(v_1,\dots,v_i)  \cdot \PP\big[\text{$\vr'$ arrives at $[v_1,\dots,v_i]_{\sim_i}$}\big]\;\mathrm{d}v_1\dots\mathrm{d}v_i\\		
	&= \sum_{i=1}^n\E{X_i \mid \tau'=i} \cdot \Pr{\tau'=i} 
	= \E{X_\tau'}.
\end{align*}\endgroup
Indeed, the second equality can be seen to hold by imagining that $X_1,\dots,X_i$ are drawn by first drawing $i$ values independently and then permuting the first $i-1$ of these values uniformly at random. The fourth equality holds by definition of $\vr'$ and by~\eqref{eq:state-distr}. This completes the proof.
\end{proof}

To further restrict the class of stopping rules from order-oblivious to value-oblivious ones we will now construct, for every order-oblivious rule $\vr$ and any $\varepsilon>0$, an infinite set $V\subseteq\N$ on which $\vr$ is $\varepsilon$-value-oblivious. The set~$V$ will depend on $\vr$ and will be obtained by starting from $\N$ and identifying smaller and smaller subsets on which the behavior of $\vr$ is more and more limited. By induction on $i\in[n]$ we will identify a set on which value-obliviousness holds with respect to the $i$th random variable.
We need the following definition.
\begin{definition}\label{def:obliv}
	Consider a $(0,n)$-stopping rule $\vr$. Let $\varepsilon>0$, $i\in[n]$, and $V\subseteq\N$. Then $\vr$ is \emph{$(\varepsilon,i)$-value-oblivious on $V$} if there exists $q\in[0,1]$ such that, for all pairwise distinct $v_1,\dots,v_i\in V$ with $v_i>\max\{v_1,\dots,v_{i-1}\}$, it holds that $r_i(v_1,\dots,v_i)\in[q-\varepsilon,q+\varepsilon)$.
\end{definition}
Note that $(\varepsilon,i)$-value-obliviousness for all $i\in[n]$ is equivalent to $\varepsilon$-value-obliviousness.
In establishing $(\varepsilon,i)$-value-obliviousness for a particular value of~$i$ we will appeal to the infinite version of Ramsey's theorem to show the existence of an appropriate set~$V$.
\begin{lemma}[\normalfont\citet{Rams30a}]  \label{lem:ramsey}
	Let $c,d\in\N$, and let~$H$ be an infinite complete $d$-uniform hypergraph whose hyperedges are colored with $c$ colors. Then there exists an infinite complete $d$-uniform sub-hypergraph of~$H$ that is monochromatic.
\end{lemma}

\begin{proof}[Proof of \autoref{lem:structure1}]
	Suppose that there exists a $(k,n)$-stopping rule with guarantee $\alpha$. By interpreting this rule as a $(0,n')$-stopping rule with $n'=k+n$, and by \autoref{lem:structure2}, there then exists a $(0,n')$-stopping rule $\vr$ that is order-oblivious and never stops on $X_1,\dots,X_k$. We fix $\varepsilon>0$ for the entire proof and show by induction on $j\in[n']$ that there exists an infinite set $S^j\subseteq\N$ such that, for all $i\in[j]$, $\vr$ is $(\varepsilon,i)$-value-oblivious on $S^j$.
	For $j=n'$, this implies that the stopping rule $\vr$ is $(\varepsilon,j)$-value-oblivious on $S^n$ for all $j\in \N$, and hence $\varepsilon$-value-oblivious on $S^n$. The claim then follows by reinterpreting $\vr$ as a $(k,n)$-stopping rule, which is possible because it never stops on $X_1,\dots,X_k$.
	
	$S^0=\N$ clearly satisfies the claim for $j=0$, and we proceed to show the claim for $j=\ell>0$ assuming that it holds for $j<\ell$.
	Note that it suffices to find an infinite set $S^{\ell}\subseteq S^{\ell-1}$ such that~$\vr$ is $(\varepsilon,\ell)$-value-oblivious on $S^{\ell}$, as the induction hypothesis then implies $(\varepsilon,i)$-value-obliviousness on $S^\ell$ as a subset of $S^i$ for all $i\in[\ell-1]$.

Toward the application of \autoref{lem:ramsey}, we construct a complete $\ell$-uniform hypergraph $H$ with vertex set $S^{\ell-1}$. Consider any set $\{v_1,\dots,v_{\ell}\}\subseteq S^{\ell-1}$ of cardinality $\ell$ such that $v_{\ell}>\max\{v_1,\dots,v_{\ell-1}\}$. Note that there exists a unique $u\in\{1,2,\dots,\lceil 1/(2\varepsilon)\rceil\}$ such that $r_{\ell}(v_1,\dots,v_{\ell})\in[(2u-1)\cdot\varepsilon-\varepsilon,(2u-1)\cdot\varepsilon+\varepsilon)$. Color the hyperedge $\{v_1,\dots,v_{\ell}\}$ of $H$ with color $u$. 

By \autoref{lem:ramsey} with $c=\lceil 1/2\varepsilon\rceil$ and $d=\ell$, there exists an infinite set of vertices that induces a complete monochromatic sub-hypergraph of $H$. Define $S^{\ell}$ to be such a set inducing a monochromatic sub-hypergraph of $H$ with color $u$. Set $q=(2u-1)\varepsilon$ and consider distinct $v_1,\dots,v_{\ell}\in S^{\ell}$ with $v_{\ell}>\max\{v_1,\dots,v_{\ell-1}\}$. Since the edge $\{v_1,\dots,v_{\ell}\}$ in $H$ has color $u$, $r_{\ell}(v_{\pi(1)},\dots,v_{\pi(\ell-1)},v_{\ell})\in[q-\varepsilon,q+\varepsilon)$ for some permutation $\pi\in\mathcal{S}_{\ell-1}$. But since $\vr$ is order-oblivious, also $r_{\ell}(v_1,\dots,v_{\ell-1},v_{\ell})\in[q-\varepsilon,q+\varepsilon)$. So $\vr$ is $(\varepsilon,\ell)$-value-oblivious on $S^{\ell}$. This completes the induction step and the proof.
\end{proof}

\subsection{Extension of the Upper Bound to $\mathbf{o(n)}$ Samples}  \label{subsec:cor}

We conclude this section by showing that even with $o(n)$ samples the guarantee of~$1/e$ is still best possible. To gain some intuition why this is true, assume that there existed an $(o(n),n)$-stopping rule $\vr$ with guarantee greater than $1/e$ by some constant. We could then obtain a $(0,n)$-stopping rule $\vr'$ that interprets, for a suitable choice of~$n'$, the first $o(n')$ values as samples and the following $n'$ values as actual values on which it may stop, and then runs $\vr$ in this setting. If we choose $n'=(1-o(1))\cdot n$, the expected maxima of $n$ and $n'$ draws from any distribution are identical up to a $(1-o(1))$ factor. The guarantee of $\vr$ would thus carry over to $\vr'$, contradicting \autoref{thm:upper}.

\begin{corollary}  \label{cor:upper_sublin}
	Let $\delta>0$ and $f:\N\rightarrow\N$ with $f(n)=o(n)$. Then there exists $n_0\in\N$ such that for any $n\geq n_0$ and any $(f(n),n)$-stopping rule with stopping time $\tau$ there exists a distribution $F$, not known to the stopping rule, such that when $X_1,\dots,X_n$ are i.i.d.\@ random variables drawn from $F$,  \vspace*{-2ex}
	\begin{align*}
		\EE[X_\tau] \leq \bigg(\frac{1}{e}+\delta\bigg) \cdot \EE[\max\{X_1,\dots,X_n\}]_{}.
	\end{align*}
\end{corollary}
\begin{proof}
	For $\delta>0$, choose $\gamma>0$ such that $(1+\gamma)/e \leq 1/e+\delta/2$ and $\gamma/(1+\gamma)\leq 1/e$. By \autoref{thm:upper_par}, there exists an $n_1$ such that for all $n\geq n_1$ and every $(\gamma\,n,n)$-stopping rule with stopping time $\tau$ there exists a distribution $F$, not known to the stopping rule, with the following property. When $X_1,\dots,X_n$ are i.i.d.\@ random variables drawn from~$F$, we have
	\begin{align*}
		\EE[X_\tau] \leq \left(\frac{1+\gamma}{e}+\frac\delta2\right) \cdot \EE[\max\{X_1,\dots,X_n\}]\leq\left(\frac{1}{e}+\delta\right) \cdot \EE[\max\{X_1,\dots,X_n\}]_{},
	\end{align*}
where the second inequality follows by our choice of $\gamma$.

Now let $n_0$ be such that $f(n)\leq\gamma\,n$ for all $n\geq n_0$. As every $(f(n),n)$-stopping rule can be interpreted as a $(\gamma\,n,n)$-stopping rule when $n\geq n_0$, the above bound for $(\gamma\,n,n)$-stopping rules applies to $(f(n),n)$-stopping rules as well when $n\geq\max\{n_0,n_1\}$. This proves the claim.
\end{proof}

\section{Linear Number of Samples}  \label{sec:linear}

The previous section has revealed a strong impossibility: even with $o(n)$ samples it is impossible to improve over the straightforward lower bound of $1/e\approx 0.368$ achieved by the well-known optimal stopping rule for the secretary problem. We proceed to show that there is a sharp phase transition when going from $o(n)$ samples to $\Omega(n)$ samples, by giving an algorithm that uses as few as $n-1$ samples and improves the lower bound from $1/e$ to $1-1/e \approx 0.632$. We also show that the bound of $1-1/e$ is in fact tight for two different classes of algorithms that share certain features of our algorithm. This illustrates that our analysis is tight and limits the types of approaches that could conceivably be used to go beyond $1-1/e$. We also show a parametric upper bound for algorithms that use $\gamma\, n$ samples for any $\gamma\geq 0$. For algorithms that use at most $n$ samples this bound is equal to $\ln(2)\approx 0.693$ and thus nearly tight.

\subsection{Warm-Up: A $\mathbf{1/2}$-Approximation with $\mathbf{n-1}$ Samples}  \label{sec:warm-up}

To gain some intuition let us first consider the natural approach to sample $n-1$ values $S_1,\dots,S_{n-1}$ from $F$ and to use the maximum of these samples as a uniform threshold for all of the random variables $X_1,\dots,X_n$, accepting the first random variable that exceeds the threshold. It is not difficult to see that the expected value we collect from any random variable $X_t$ conditioned on stopping at that random variable is at least $\E{\max\{X_1,\dots,X_n\}}$, since under this condition~$X_t$ is the maximum of at least~$n$ i.i.d.\@ random variables. We can thus understand the approximation guarantee provided by this approach by understanding the probability that it stops on some random variable. It turns out that this probability, and hence the approximation guarantee, is $1/2+1/(4n-2)$. A detailed analysis of the approach is provided for completeness in \autoref{app:one-half}. 

\subsection{A $\mathbf{(1-1/e)}$-Approximation with $\mathbf{n-1}$ Samples}  \label{subsec:1-1e}

We proceed to show that it is indeed possible to obtain an improved bound of $1-(1-1/n)^n\geq 1-1/e\approx 0.632$ with just $n-1$ samples. 
Our algorithm improves over the naive approach that obtains a factor of~$1/2$ by increasing the probability that we stop at all, while maintaining the property that the expected value that we collect when we do stop is at least $\E{\max\{X_1,\dots,X_n\}}$.
\begin{theorem}\label{thm:fresh-looking-samples}
	Let $X_1,X_2,\dots,X_n$ be i.i.d.\@ random variables from an unknown distribution $F$. Then there exists an $(n-1,n)$-stopping-rule with stopping time~$\tau$ such that
\begin{align*}
	\E{X_\tau} = \left(1-\left(1-\frac{1}{n}\right)^n\right) \cdot \E{\max\{X_1, \dots, X_n\}}.
\end{align*}
\end{theorem}

Note that a guarantee of $1-1/e-\varepsilon$ with $O_\varepsilon(n)$ samples follows from a result of \citet{EhsaniHKS18} by observing that $O_\varepsilon(n)$ samples provide a sufficiently good approximation to the $1/e$-quantile of the distribution of $\max\{X_1,\dots,X_n\}$. Here we take a different route that yields the bound \emph{exactly} and that, more importantly, can be developed further to work with only $n-1$ samples.

Suppose we were given access to $n(n-1) \in \Theta(n^2)$ samples. Then we could partition the $n(n-1)$ samples into~$n$ sets of size $n-1$ each, and use the maximum of the $i$th set as a threshold for the $i$th random variable. Upon acceptance of any random variable, that random variable would have a value equal to the expected maximum of $n$ i.i.d.\@ random variables, which is equal to $\E{\max\{X_1,\dots,X_n\}}$. Conditioned on reaching the $i$th random variable it would be accepted with probability $1/n$, for an overall probability of acceptance of $\sum_{i=1}^{n}(1-1/n)^{i-1} \cdot 1/n=1-(1-1/n)^n$.

\begin{algorithm}[t]
\SetInd{0.5em}{1em}
\DontPrintSemicolon
\KwData{Sequence of i.i.d.~random variables $X_1,\dots,X_n$ sampled from an unknown distribution~$F$, sample access to $F$} %
\KwResult{Stopping time $\tau$}
  $S_1, \dots, S_{n-1} \longleftarrow \text{$n-1$ independent samples from $F$}$\;
  $S \longleftarrow \{S_1, \dots, S_{n-1}\}$ \;
  \For{$t = 1, \dots, n$}{
    \lIf{$X_t \geq \max S$}{
		\Return $t$
    }
    \Else{
	$S \longleftarrow \text{random subset of size $n-1$ of $\{S_1, \dots, S_{n-1}, X_1, \dots, X_t\}$}$\\
    }
    }
    \Return $n+1$
\caption{Fresh-looking samples}
\label{alg:fresh-looking-samples}
\end{algorithm}

\autoref{alg:fresh-looking-samples} mimics this approach, but instead of using $n-1$ fresh samples for each of the $n$ random variables it constructs $n-1$ fresh-looking samples for each of the $n$ random variables from a \emph{single} set $\{S_1,\dots,S_{n-1}\}$ of $n-1$ samples. The algorithms starts by drawing $n-1$ samples $S_1,\dots,S_{n-1}$. Then, for each time step $t$, it compares the current random variable $X_t$ to the maximum $\max S$ of a random subset~$S$ of size $n-1$ of the set $\{S_1,\dots,S_{n-1},X_1,\dots,X_{t-1}\}$ containing the initial samples and the random variables seen previously. If $X_t\geq\max S$, the algorithm accepts $X_t$ and stops. Otherwise, it continues to the next random variable.

The key ingredient in our analysis is the following lemma, which concerns the distribution of the unordered set of values seen before step $t$ under the condition that the algorithm has reached that step.

\begin{lemma}\label{lem:independence}
	If \autoref{alg:fresh-looking-samples} arrives at step~$t$, the distribution of the set $\{S_1,\dots,S_{n-1}, X_1,\dots,\linebreak[1]X_{t-1}\}$ is identical to the distribution of a set of $n+t-2$ fresh samples from $F$.
\end{lemma}
\begin{proof}
	We show the claim by induction on~$t$, and start by observing that it clearly holds for $t=1$ as $\{S_1, \dots, S_{n-1}\}$ \emph{is} a set of $n-1$ fresh samples from $F$.
	
	Now assume that the claim holds for $t=1,\dots,t^\star-1$. Then, by the induction hypothesis and under the condition that the algorithm arrives at step $t^\star-1$, the set $\mathcal{T}=\{S_1,\dots,S_{n-1},X_1,\dots,X_{t^\star-2}\}$ has the same distribution as a set of $n+t^\star-3$ fresh samples from $F$. We now consider the set $\mathcal{T}' = \{S_1,\dots,S_{n-1},X_1,\dots,X_{t^\star-1}\}$, which additionally includes $X_{t^\star-1}$, and claim that conditioned on arriving at step $t^\star$ this set is distributed like a set of $n+t^\star-2$ fresh samples. Note that the latter is true \emph{before} the decision to stop or not to stop at step $t^\star-1$ is taken. To show the claim we will argue that the decision of the algorithm to stop or not to stop at step $t^\star-1$ does not depend on the realization of the set $\mathcal{T}'$.
	
	Fix any realization $\{y_1,\dots,y_{n+t^\star-2}\}$ of $\mathcal{T'}$, and assume that the algorithm has arrived at step $t^\star-1$. Since $F$ is continuous we may assume that the values $y_1,\dots,y_{n+t^\star-2}$ are pairwise distinct, and without loss of generality that $y_1<y_2<\dots<y_{n+t^*-2}$. Since $\mathcal{T}$ is distributed like a set of fresh samples and $X_{t^\star-1}$ is a fresh sample, it must be the case that $X_{t^\star-1}$ is distributed uniformly over $\{y_1,\dots,y_{n+t^\star-2}\}$ and the elements of $\mathcal{T}$ are equal to the remaining values in $\{y_1,\dots,y_{n+t^\star-2}\}$. The algorithm now stops at step $t^\star-1$ if $X_{t^\star-1}=y_{n+t^\star-2}$, and this happens with probability $1/n$.
\end{proof}

\begin{proof}[Proof of \autoref{thm:fresh-looking-samples}]
	The value $\E{X_\tau}$ obtained by \autoref{alg:fresh-looking-samples} can be written by summing over all possible stopping times $t=1,\dots,n$ the product of the probability of stopping at $X_t$ and the expectation of $X_t$ upon stopping. Writing $S^t$ for the random subset used in step $t$, we thus have
\begin{align}
	\E{X_\tau}\; & = \sum_{t=1}^{n} \bigg(\Pr{X_t \geq \max\{S^t\} \wedge X_j < \max\{S^j\} \text{ for $j < t$}} \notag \\[-2ex]
	&\hspace*{80pt}\cdot\; \E{X_t \mid X_t \geq \max\{S^t\} \wedge X_j < \max\{S^j\} \text{ for $j < t$}}\bigg). \label{eq:alg-formula}
\end{align}

For any $t\in\{1,\dots,n\}$,
\begin{align}
	&\Pr{X_t \geq \max\{S^t\} \wedge X_j < \max\{S^j\} \text{ for $j < t$}} \notag \\[1ex]
	&\hspace{110pt} =\;\Pr{X_t \geq \max\{S^t\} \mid X_j < \max\{S^j\} \text{ for $j < t$}} \notag \\[6pt]
	&\hspace*{130pt} \;\cdot \prod_{\ell < t} \Pr{X_\ell < \max\{S^\ell\} \mid X_j < \max\{S^j\}  \text{ for $j < \ell$}} \notag\\
	&\hspace{110pt} =\;\frac{1}{n} \left(1-\frac{1}{n}\right)^{t-1}, \label{eq:prob-formula}
\end{align}
where the first equality can be obtained by repeated application of the definition of conditional probabilities and the second equality follows from \autoref{lem:independence}.

Denoting by $\{T_1,\dots,T_{n-1}\}$ a set of fresh samples from $F$, we claim that
\begin{align}
	&\E{X_t \mid X_t \geq \max S^t \wedge X_j < \max S^j \text{ for $j < t$}} \notag \\[6pt]
	&\hspace{110pt}=\; \E{X_t \mid X_t \geq \max\{T_1, \dots, T_{n-1}\} \wedge X_j < \max S^j \text{ for $j < t$}}  \notag\\[6pt]
	&\hspace{110pt}=\; \E{X_t \mid X_t \geq \max\{T_1, \dots, T_{n-1}\}}  \notag\\[6pt]
	&\hspace{110pt}=\; \E{\max\{X_1, \dots, X_n\}}. \label{eq:exp-formula}
\end{align}
Indeed the first equality holds because, under the condition that $X_j<\max S^j$ for $j<t$ and by \autoref{lem:independence}, $S^t$ is distributed like $\{T_1,\dots,T_{n-1}\}$. The second equality holds because $X_t$ itself is independent of whether $X_j<\max S^j$ for $j<t$, and the third equality because $X_t$ is distributed like a fresh sample.

By substituting~\eqref{eq:prob-formula} and~\eqref{eq:exp-formula} into~\eqref{eq:alg-formula} we obtain
\begin{align*}
	\E{X_\tau} 
	&= \sum_{t=1}^{n} \left(\frac{1}{n} \left(1-\frac{1}{n}\right)^{t-1} \right)\cdot \E{\max\{X_1, \dots, X_n\}} \\
	&= \left(1-\left(1-\frac{1}{n}\right)^n\right)\cdot \E{\max\{X_1, \dots, X_n\}},
\end{align*}
as claimed.
\end{proof}

\subsection{Going Beyond $\mathbf{1-1/e}$}\label{subsec:beyond}

We proceed to show an upper bound of $1-1/e$ for two different classes of algorithms that share certain characteristics of \autoref{alg:fresh-looking-samples}. This shows that our analysis of \autoref{alg:fresh-looking-samples} is tight and limits the class of algorithms that could conceivably improve on the guarantee of $1-1/e$.

Algorithms in the first class, upon reaching the $i$th random variable, stop at this random variable with a probability that is independent of~$i$. This is the case for \autoref{alg:fresh-looking-samples}, which by \autoref{lem:independence} 
stops with probability $1/n$ upon reaching a particular random variable.
The upper bound we obtain applies even in the case where the distribution~$F$ is known, and to stopping rules that like \autoref{alg:fresh-looking-samples} use dependent thresholds.
\begin{proposition}  \label{prop:1-1e-best-possible}
	Let $\delta>0$. Then there exists $n\in\N$ and a distribution $F$ such that for any stopping time~$\tau$ for which $\Pr{\tau=i\mid\tau>i-1}$ is independent of~$i$,
	\[
		\E{X_\tau} \leq \left(1-\frac{1}{e}+\delta\right) \cdot \E{\max\{X_1,\dots,X_n\}}.
	\]
\end{proposition}
\begin{proof}
	For $n\geq 3$ and $i\in[n]$, let
\[
	X_i=\begin{cases}
		\frac{\sqrt{n}}{e-2} & \text{with probability $\frac{1}{n^{3/2}}$,} \\
		1 & \text{with probability $\frac{1}{\sqrt{n}}$,} \\
		0 & \text{otherwise.}
	\end{cases}
\]

	We begin by bounding $\EE[\max\{X_1, \dots, X_n\}]$ from below. For every $\varepsilon>0$ there exists $m\in\N$ such that for all $n \geq m$,
\begin{align*}
	&\E{\max\{X_1,\ldots,X_n\}} \\
	&\qquad = \Pr{X_i = \frac{\sqrt{n}}{e-2} \text{ for some $i$}} \cdot \frac{\sqrt{n}}{e-2} \\[1ex]
	&\qquad \phantom{{}=} + \Pr{X_i < \frac{\sqrt{n}}{e-2} \text{ for all $i$}} \cdot \Pr{X_i = 1 \text{ for some $i$} \mid X_i < \frac{\sqrt{n}}{e-2} \text{ for all $i$}} \cdot 1 \\
	&\qquad = \Bigl(1-\Bigl(1-\frac{1}{n^{3/2}}\Bigr)^n\Bigr) \cdot \frac{\sqrt{n}}{e-2} + \Bigl(1-\frac{1}{n^{3/2}}\Bigr)^n \cdot \Bigl(1-\Bigl(\frac{1-\frac{1}{\sqrt{n}}-\frac{1}{n^{3/2}}}{1-\frac{1}{n^{3/2}}}\Bigr)^n\Bigr) \cdot 1 \\
	&\qquad \ge \frac{1}{e-2} + 1 - \varepsilon
	= \frac{e-1}{e-2} - \varepsilon	.
\end{align*}
\vspace*{-6pt}

In bounding $\E{X_{\tau}}$ from above, we can restrict attention to stopping rules that always accept a value of $\frac{\sqrt{n}}{e-2}$ and never accept a value of~$0$. Given the property that $\Pr{\tau=i\mid\tau>i-1}$ is independent of~$i$, any such stopping rule is characterized by the probability with which it accepts a value of~$1$. Denoting this probability, which may depend on $n$, by $q_n$, and the corresponding stopping time by $\tau_{q_n}$,
\begin{align*}
	\E{X_{\tau_{q_n}}}
	&= \sum_{i=1}^{n} \Pr{\tau_{q_n} > i-1} \cdot \Pr{\tau_{q_n} = i \mid \tau_{q_n} > i-1} \cdot \E{X_i \mid \tau_{q_n} = i} \\[.5ex]
	&= \sum_{i = 1}^{n} \biggl(1-\frac{q_n}{\sqrt{n}}-\frac{1}{n^{3/2}}\biggr)^{i-1} \biggl( \frac{q_n}{\sqrt{n}} \cdot 1 + \frac{1}{n^{3/2}} \cdot \frac{\sqrt{n}}{e-2}\biggr) \\[.5ex]
	&= \biggl(1-\biggl(1-\frac{q_n}{\sqrt{n}}-\frac{1}{n^{3/2}}\biggr)^n\biggr) \biggl(\frac{q_n/\sqrt{n}}{q_n/\sqrt{n}+1/n^{3/2}}\cdot 1 + \frac{1/n^{3/2}}{q_n/\sqrt{n}+1/n^{3/2}}\cdot \frac{\sqrt{n}}{e-2}\biggr) \\[1ex]
	&= \biggl(1-\biggl(1-\frac{q_n\sqrt{n}+1/\sqrt{n}}{n}\biggr)^n\biggr) \frac{q_n\sqrt{n}+1/(e-2)}{q_n\sqrt{n}+1/\sqrt{n}}.
\end{align*}
We now distinguish three cases depending on the limit behavior of $q_n\sqrt{n}$.

If $\lim\sup_{n\rightarrow\infty}q_n\sqrt{n}=\infty$, then for every $\varepsilon>0$ and $m\in\N$, there exists $n\geq m$ such that
\begin{align*}
	\E{X_{\tau_{q_n}}}
	&= \biggl(1-\biggl(1-\frac{q_n\sqrt{n}+1/\sqrt{n}}{n}\biggr)^n\biggr) \frac{q_n\sqrt{n}+1/(e-2)}{q_n\sqrt{n}+1/\sqrt{n}} \\[1ex]
	&\leq \frac{q_n\sqrt{n}+1/(e-2)}{q_n\sqrt{n}+1/\sqrt{n}} \\[1ex]
	&\leq 1 + \varepsilon.
\end{align*}
Indeed, the first inequality holds because $(1-(1-\frac{1}{n}(q_n\sqrt{n}+1/\sqrt{n}))^n)\leq1$, and the second inequality because $\lim\sup_{n\rightarrow\infty} q_n \sqrt{n} = \infty$.

If $\lim\inf_{n\rightarrow\infty}q_n\sqrt{n}=0$, then for every $\varepsilon>0$ and every $m\in\N$, there exists $n\geq m$ such that
\begin{align*}
	\E{X_{\tau_{q_n}}}
	&= \biggl(1-\biggl(1-\frac{q_n\sqrt{n}+1/\sqrt{n}}{n}\biggr)^n\biggr) \frac{q_n\sqrt{n}+1/(e-2)}{q_n\sqrt{n}+1/\sqrt{n}} \\[.5ex]
	&\leq (q_n \sqrt{n} + 1/\sqrt{n}) \frac{q_n\sqrt{n}+1/(e-2)}{q_n\sqrt{n}+1/\sqrt{n}} \\[.5ex]
	&\leq \frac{1}{e-2} + \varepsilon .
\end{align*}
For the first inequality we have used that for $y=q_n\sqrt{n}+1/\sqrt{n}$, and by Bernoulli's inequality, $1-(1-y/n)^n\leq y$ provided that $y/n\leq 1$, which is satisfied because $q_n\leq 1$ and $n\geq 3$. The second inequality holds because $\lim\inf_{n\rightarrow\infty}q_n\sqrt{n}=0$.

Finally, if $\lim\sup_{n \rightarrow \infty}q_n\sqrt{n}<\infty$ and $\lim\inf_{n \rightarrow\infty}q_n\sqrt{n}>0$, then there exists constants $c_1,c_2$ with $0<c_1\leq c_2$ and infinitely many values of $n$ such that $q_n\sqrt{n}+1/\sqrt{n}\in[c_1,c_2]$. Then, for every $\varepsilon>0$ and $m\in\N$, and for $\varepsilon'=\varepsilon\cdot(c_2+1/(e-2))/c_1$, there exists $n$ with $n\geq m$ and $q_n\sqrt{n} 1/\sqrt{n}\in[c_1,c_2]$ such that
\begin{align*}
	\E{X_{\tau_{q_n}}}
	&= \biggl(1-\biggl(1-\frac{q_n\sqrt{n}+1/\sqrt{n}}{n}\biggr)^n\biggr) \frac{q_n\sqrt{n}+1/(e-2)}{q_n\sqrt{n}+1/\sqrt{n}} \\[1ex]
	&\leq \bigl(1-e^{-(q_n\sqrt{n}+1/\sqrt{n})}+\varepsilon'\bigr) \frac{q_n\sqrt{n}+1/(e-2)}{q_n\sqrt{n}+1/\sqrt{n}} \\[1ex]
	&\leq \bigl(1-e^{-(q_n\sqrt{n}+1/\sqrt{n})}\bigr) \frac{q_n\sqrt{n}+1/(e-2)}{q_n\sqrt{n}+1/\sqrt{n}} + \varepsilon \\[1ex]
	&\leq \max_{x \geq 0} \biggl(\bigl(1-e^{-x}\bigl) \frac{x+1/(e-2)}{x} \biggr)+ \varepsilon \\[1ex]
	&\leq \frac{(e-1)^2}{(e-2)e} + \varepsilon.
\end{align*}
For the first inequality we have used that for every $\varepsilon>0$ there exists $m\in\N$ such that for all $n\geq m$ and all $x\in[c_1,c_2]$ it holds that $1-(1-x/n)^n\leq 1-e^{-x}+\varepsilon$. The last inequality holds because the maximum of $(1-e^{-x}) \cdot (x+1/(e-2))/x$ is attained at $x=1$, where it is equal to $(e-1)^2/((e-2)e)$.

For every stopping time $\tau$ such that $\Pr{\tau=i\mid\tau>i-1}$ is independent of~$i$, and for every $\varepsilon\geq 0$ and $m\in\N$, there thus exists $n\geq m$ such that
\begin{align*}
	\E{X_{\tau}} \leq \max \biggl\{ 1 + \varepsilon, \frac{1}{e-2} + \varepsilon, \frac{(e-1)^2}{(e-2)e} + \varepsilon \biggr\} = \frac{(e-1)^2}{(e-2)e} + \varepsilon.
\end{align*}

Let $f:\R\to\R$ such that for every $\varepsilon\geq 0$,
\[
	f(\varepsilon) = \frac{\frac{(e-1)^2}{(e-2)e}+\varepsilon}{\frac{e-1}{e-2}-\varepsilon}.
\]
Then, for every $\varepsilon>0$, there exists $n\in N$ such that
\begin{align*}
	\E{X_{\tau}} / \E{\max\{X_1,\ldots,X_n\}} &\leq f(\varepsilon) .
\end{align*}
Since $f$ is continuous and $\lim_{\varepsilon\to 0}=1-\frac{1}{e}$ there exists, for every $\delta>0$, a value $\varepsilon>0$ such that $f(\varepsilon)\leq 1-\frac{1}{e}+\delta$, and thus $n\in\N$ such that
\begin{align*}
	\E{X_{\tau}} / \E{\max\{X_1,\ldots,X_n\}} &\leq 1-\frac{1}{e}+\delta,
\end{align*}
as claimed.
\end{proof}

Algorithms in the second class have access to $n-1$ samples $S_1,\dots,S_{n-1}$ from the underlying distribution and satisfy the following two natural conditions: (i)~if the value of the first random variable~$X_1$ is greater than all $n-1$ samples, they stop; and (ii)~conditioned on reaching~$X_i$, their probability of stopping at~$X_i$ is non-decreasing in~$i$. It is again easily verified that \autoref{alg:fresh-looking-samples} belongs to this class.
\begin{proposition}  \label{prop:alternative-upper-bound}
Let $\delta>0$. Then there exists $n\in\N$ and a distribution~$F$ such that for any $(n-1,n)$-stopping rule with stopping time~$\tau$ that satisfies Conditions~(i) and~(ii),
\[
	\E{X_\tau} \leq \left(1-\frac{1}{e}+\delta\right) \cdot \E{\max\{X_1,\dots,X_n\}}.
\]
\end{proposition}
\begin{proof}
Let $\varepsilon>0$. Let $Y_i$ be distributed uniformly on $[0,\varepsilon]$, and let $X_i=1+Y_i$ with probability $1/n^2$ and $X_i=Y_i$ otherwise. Consider a stopping rule $\vr$ with stopping time $\tau$ that satisfies Conditions~(i) and~(ii). Then $\Pr{\tau=1}\geq 1/n$ and $\Pr{\tau=i\mid\tau\geq i}\geq\Pr{\tau=1}\geq 1/n$. Moreover, it is easy to see by induction that $\Pr{\tau\geq i} \leq (1-1/n)^{i-1}$. Indeed, $\Pr{\tau\geq 1}=1$; for $i>1$, and assuming that the claim holds for $i-1$,
\[
	\Pr{\tau \geq i} = \Pr{\tau \geq i-1} \cdot \Pr{\tau \neq i-1 \mid \tau \geq i-1} \leq (1-1/n)^{i-2} \cdot (1-1/n) = (1-1/n)^i.
\]
It follows that
\[
	\E{X_\tau}\leq\sum_{i = 1}^{n} \Pr{\tau \geq i} \cdot \frac{1}{n^2} + \varepsilon \leq \sum_{i=1}^{n} \biggl(1-\frac{1}{n}\biggr)^{i-1} \frac{1}{n^2} + \varepsilon = \frac{1}{n} \biggl(1-\biggl(1-\frac{1}{n}\biggr)^n\biggr) + \varepsilon.
 \]
On the other hand
\[
	\E{\max\{X_1, \dots, X_n\}} \geq 1-\biggl(1-\frac{1}{n^2}\biggr)^n,
\]
and therefore
\[
	\frac{\E{X_\tau}}{\E{\max\{X_1, \dots, X_n\}}} \leq \frac{\frac{1}{n} \Bigl(1-\bigl(1-\frac{1}{n}\bigr)^n\Bigr)+\varepsilon}{1-\bigl(1-\frac{1}{n^2}\bigr)^n}.
\]
The right-hand side tends to $1-1/e$ as $n\rightarrow\infty$ and $\varepsilon\to 0$, so for every $\delta>0$ there exists $n\in\N$ such that $\E{X_\tau}\leq (1-1/e+\delta) \cdot \E{\max\{X_1, \dots, X_n\}}$.
\end{proof}

\subsection{A Parametric Lower Bound}

The lower bound of \autoref{thm:fresh-looking-samples} can be generalized to a situation with $\gamma n$ samples when $\gamma\in[0,1]$. The idea is to interpret some of the values $X_1,\dots,X_n$ as samples, so that the number of remaining values equals the number of samples and \autoref{alg:fresh-looking-samples} can be applied.
\begin{corollary}\label{cor:parametric-lower}
	Let $X_1,X_2,\dots,X_n$ be i.i.d.\@ random variables from an unknown distribution $F$. Let $\gamma\in[0,1]$ such that $\gamma n+n$ is an even number. Then there exists an $(\gamma n,n)$-stopping-rule with stopping time~$\tau$ such that
\begin{align*}
	\E{X_\tau} \geq \frac{1+\gamma}{2}\cdot\left(1-\frac 1e\right) \cdot \E{\max\{X_1, \dots, X_n\}}.
\end{align*}
\end{corollary}
\begin{proof}
	Let $n'=\frac{1+\gamma}{2}n$, and note that $n'\in\N$. Define $S'_i=S_i$ for all $i\in[\gamma n]$, $S'_{\gamma n+i}=X_{i}$ for all $i\in[n'-\gamma n]$, and $X'_i=X_{n'-\gamma n+i}$ for all $i\in[n']$. Note that $X'_{n'}=X_n$, so this assignment is well-defined. We use \autoref{alg:fresh-looking-samples} with stopping time $\tau$ on $X'_1,\dots,X'_{n'}$ with samples $S'_1,\dots,S'_{n'-1}$. Then by applying \autoref{thm:fresh-looking-samples} we get 
	\[\E{X'_\tau}\geq \left(1-\frac 1e\right) \cdot \E{\max\{X'_1, \dots, X'_{n'}\}}\geq \frac{1+\gamma}{2}\cdot\left(1-\frac 1e\right) \cdot \E{\max\{X_1, \dots, X_n\}},\]
	as claimed.
\end{proof}

\subsection{A Parametric Upper Bound} \label{subsec:parametric}

\begin{figure}
\centering
\hspace*{-35pt}
\begin{tikzpicture}
\begin{axis}[ width=0.55\textwidth,height=0.42\textwidth,xlabel=$\gamma$, axis x line = left, axis y line = left, xmin = 0, xmax = 1.32, ymin = 0, ymax = 0.8, font=\small, line width = 1pt] 
\addplot[no marks,very thick,domain=0:0.58197,dashed] {(1+x)/exp(1)}; 
\addplot[no marks,very thick,domain=0.58197:1.32,dashed] {-x*ln(x/(1+x))};
\addplot[no marks,very thick,domain=0:0.164] {0.3678};
\addplot[no marks,very thick,domain=0.164:1] {0.316*x+0.316};
\addplot[no marks,very thick,domain=1:1.32] {0.6321};
\end{axis}
\end{tikzpicture}
\caption{Visualization of the parametric lower bound (solid) and the parametric upper bound (dashed).}\label{fig:parametric}
\end{figure}

While an improvement over the bound of $1-1/e\approx 0.632$ remains conceivable via more complicated stopping rules, such an improvement cannot go beyond $\ln(2)\approx 0.693$. This is a consequence of the following generalization of \autoref{thm:upper}, which provides a parametric upper bound for stopping rules that have access to $\gamma\,n$ samples for some $\gamma\geq 0$. We prove this result by generalizing the proof of \autoref{thm:upper}, and bounding the performance of the algorithm by bounding the probability that it accepts the maximum of the entire sequence of $(1+\gamma)n$ values.
\begin{theorem}
	\label{thm:upper_par}
		Let $\delta>0$, $\gamma\in\mathbb{Q}_{+}$. Then there exists $n_0\in\N$ such that for any $n\geq n_0$ and any $(\gamma\, n,n)$-stopping rule with stopping time $\tau$ there exists a distribution $F$, not known to the stopping rule, with the following property. When $X_1,\dots,X_n$ are i.i.d.\@ random variables drawn from~$F$,
	\begin{align*}
		\EE[X_\tau] \leq (b(\gamma)+\delta) \cdot \EE[\max\{X_1,\dots,X_n\}]_{},
	\end{align*}
	where
	\[ 
	b(\gamma)=\begin{cases}
	\frac{1+\gamma}e & \text{if }\frac1e\geq \frac\gamma{1+\gamma}, \\
	-\gamma\cdot\log\frac{\gamma}{1+\gamma} & \text{otherwise.}
\end{cases}
\]
\end{theorem}
\begin{proof}
We will restrict attention to $n\in\N$ such that $\gamma n\in\N$. It suffices to show that the guarantee of any $(\gamma n,n)$-stopping rule is bounded from above by $b(\gamma)+o(1)$, where implicitly $n\rightarrow\infty$. To this end, consider an arbitrary $(\gamma n,n)$-stopping rule with guarantee $\alpha$.
Let $\varepsilon=1/n^2$. By \autoref{lem:structure1} there then exists an infinite set $V\subseteq\N$ on which~$\vr$ is $\varepsilon$-value-oblivious. Let $v_1,\dots,v_{n^3},u\in V$ be pairwise distinct such that $u\geq n^3\max\{v_1,\dots,v_{n^3}\}$. For each $i\in[n]$, let
\[
	X_i = \begin{cases}
		v_1 & \text{with probability $\frac{1}{n^3}\cdot(1-\frac{1}{n^2})$,} \\
		\vdots & \\
		v_{n^3} &\text{with probability $\frac{1}{n^3}\cdot(1-\frac{1}{n^2})$,} \\
		u &\text{with probability $\frac{1}{n^2}$.}
	\end{cases}
\]

We proceed to bound $\EE[\max\{X_1,\dots,X_n\}]$ from below and $\EE[X_\tau]$ from above. Let $X_{(i)}$ denote $i$th order statistic of $X_1,\dots,X_n$, such that $X_{(n)}=\max\{X_1,\dots,X_n\}$. Then
\begin{equation}\label{eq:thm4-opt}
	\EE[\max\{X_1,\dots,X_n\}] \geq \PP[X_{(n)}=u]\cdot u=\frac{1-o(1)}{n}\cdot u.
\end{equation}
For $i\in[(1+\gamma)n]$, let
\begin{align*}
	R_i = \begin{cases}
		S_i & \text{if $i\leq \gamma n$,} \\
		X_{i-\gamma n} & \text{otherwise,}
	\end{cases}
\end{align*} and let $\tau'=\tau+\gamma n$ be the stopping time of~$\vr$ viewed as a $(0,(1+\gamma)n)$-stopping rule on $R_1,\dots,R_{(1+\gamma)\cdot n}$. Then, analogously to the proof of \autoref{thm:upper} 
\begin{align}
	\EE[X_\tau] &= \PP[R_{((1+\gamma)\cdot n)}=u\wedge R_{((1+\gamma)\cdot n-1)}\neq u]\cdot\EE[R_{\tau'}\mid R_{((1+\gamma)\cdot n)}=u\wedge R_{((1+\gamma)\cdot n-1)}\neq u] \nonumber \\[.5ex]
	& \phantom{={}} + \PP[R_{((1+\gamma)\cdot n)}=u\wedge R_{((1+\gamma)\cdot n-1)}=u]\cdot\EE[R_{\tau'}\mid R_{((1+\gamma)\cdot n)}=u\wedge R_{((1+\gamma)\cdot n-1)}=u] \nonumber \\[.5ex]
	& \phantom{={}} +\PP[R_{((1+\gamma)\cdot n)}\neq u]\cdot\EE[R_{\tau'}\mid R_{((1+\gamma)\cdot n)}\neq u] \nonumber \\[1ex]
	&\leq \frac{1+\gamma}{n}\cdot\Bigl(\PP[R_{\tau'}=R_{((1+\gamma)\cdot n)}\mid R_{((1+\gamma)\cdot n)}=u\wedge R_{((1+\gamma)\cdot n-1)}\neq u]\cdot u \nonumber\\[-1ex]
	& \phantom{=\frac{1}{n}\cdot\big({}}+\PP[R_{\tau'}\neq R_{((1+\gamma)\cdot n)}\mid R_{((1+\gamma)\cdot n)}=u\wedge R_{((1+\gamma)\cdot n-1)}\neq u]\cdot O(n^{-3})\cdot u\Bigr) \nonumber\\
	& \phantom{={}} + O(n^{-2})\cdot u+1\cdot O(n^{-3})\cdot u\nonumber\\[1ex]
 	&\leq \frac{1+\gamma}{n}\cdot \PP[R_{\tau'}=R_{((1+\gamma)\cdot n)}\mid R_{((1+\gamma)\cdot n)}=u\wedge R_{((1+\gamma)\cdot n-1)}\neq u]\cdot u +o\left(\frac 1n\right)\cdot u\nonumber \\
  & \leq\frac{1+\gamma}{n}\cdot \PP[R_{\tau'}=R_{((1+\gamma)\cdot n)}\mid R_{((1+\gamma)\cdot n)}=u\wedge R_1,\dots,R_{(1+\gamma)\cdot n}\text{ are distinct}]\cdot u \nonumber\\
	& \phantom{={}} + o\left(\frac 1n\right)\cdot u.\label{eq:thm4-alg}
\end{align}

Given~\eqref{eq:thm4-opt} and~\eqref{eq:thm4-alg}, to show $\alpha\leq b(\gamma)+o(1)$, it suffices to show that
\begin{equation}\label{eq:thm4-goal}
	\PP[R_{\tau'}=R_{((1+\gamma)\cdot n)}\mid R_{((1+\gamma)\cdot n)}=u\wedge R_1,\dots,R_{(1+\gamma)\cdot n}\text{ are distinct}]\leq b(\gamma)/(1+\gamma)+o(1).
\end{equation}
Note that in the event where $R_{((1+\gamma)\cdot n)}=u$ and $R_1,\dots,R_{(1+\gamma)\cdot n}\text{ are distinct}$, the relative ranks of $R_1,\dots,R_{(1+\gamma)\cdot n}$ are distributed uniformly at random. The optimal stopping rule for accepting the value with the largest relative rank is known to set, for some $x\in[0,1]$, $q_i=0$ for all $i<x\cdot(1+\gamma)\cdot n$ and $q_i=1$ for all $i\geq x\cdot(1+\gamma)\cdot n$~\citep{GiMo66a}. Then, for any $(0,(1+\gamma)\cdot n)$-stopping rule $\hat{\vr}$ with stopping time $\hat{\tau}$ that does not accept any of the values $X_1,\dots,X_{\gamma n}$, 
\[\PP[X_{\hat{\tau}}=R_{((1+\gamma)\cdot n)}\mid R_{((1+\gamma)\cdot n)}=u\wedge R_1,\dots,R_{(1+\gamma)\cdot n}\text{ are distinct}]=-x\cdot\log x+o(1),
\]
which subject to $x\geq \gamma/(1+\gamma)$ is maximized for
\[ x=\max\left\{\frac1e,\frac\gamma{1+\gamma}\right\}. \] 
Thus
\begin{equation}\label{eq:thm4-secretary}
	\PP[X_{\hat{\tau}}=R_{((1+\gamma)\cdot n)}\mid R_{((1+\gamma)\cdot n)}=u\wedge R_1,\dots,R_{(1+\gamma)\cdot n}\text{ are distinct}]\leq b(\gamma)/(1+\gamma)+o(1),
\end{equation}
where 
\[ b(\gamma)=\begin{cases}
	\frac{1+\gamma}e & \text{if }\frac1e\geq \frac\gamma{1+\gamma}, \\
	-\gamma\cdot\log\frac{\gamma}{1+\gamma} & \text{otherwise.}
\end{cases} \]
Analogously to the proof of \autoref{thm:upper}, $\vr$ can be coupled with a $(0,(1+\gamma)\cdot n)$-stopping rule $\hat{\vr}$ as above such that
\begin{multline}
	\PP[X_{\tau'}=R_{((1+\gamma)\cdot n)}\mid R_{((1+\gamma)\cdot n)}=u\wedge R_1,\dots,R_{(1+\gamma)\cdot n}\text{ are distinct}] 
	\leq \\[1ex] \PP[X_{\hat{\tau}}=R_{((1+\gamma)\cdot n)}\mid R_{((1+\gamma)\cdot n)}=u\wedge R_1,\dots,R_{(1+\gamma)\cdot n}\text{ are distinct}]+n\varepsilon.\label{eq:thm4-coupling}
\end{multline}
Then substituting \eqref{eq:thm4-coupling} into \eqref{eq:thm4-secretary} yields \eqref{eq:thm4-goal}, completing the proof.
\end{proof}

A visualization of the upper bound of \autoref{thm:upper_par} and the lower bound of \autoref{cor:parametric-lower} is shown in \autoref{fig:parametric}. Note that the function~$b$ defined in \autoref{thm:upper_par} is continuous and that $b(1)=\ln(2)\approx 0.693$. Moreover, $b$ meets the tight bound of approximately~$0.745$ due to \citet{CorreaFHOV17}, which implies an upper bound in the setting where the distribution is unknown, at $\gamma\approx 1.32$.

\section{Superlinear Number of Samples}  \label{sec:745}

Our final result is that we can in fact match the optimal guarantee achievable by a stopping rule that knows the distribution, up to any $\varepsilon>0$, if we have access to $O_\varepsilon(n^2)$ samples.

\begin{theorem}\label{thm:quadratic}
For every $\varepsilon>0$, there exists an $n_\varepsilon\in\N$ such that the following holds for all $n\geq n_\varepsilon$. Let $X_1, \dots, X_n$ be i.i.d.\@ random variables drawn from an unknown distribution $F$. Then there exists an algorithm for choosing a stopping time $\tilde{\tau}$ that uses $O(n^2)$ samples from the same distribution with
\[
	\E{X_{\tilde{\tau}}} \geq (\beta^{-1}-\varepsilon) \cdot \E{\max\{X_1, \dots, X_n\}}, 
\]
where $\beta^{-1}\approx 0.745$ is the guarantee shown by \citet{CorreaFHOV17}.
\end{theorem}

As our algorithm is related to that of \citet{CorreaFHOV17}, we first recall how that algorithm works. It computes a decreasing sequence $x_1,x_2,\dots,x_n$. As $n\rightarrow\infty$, it can be shown that $x_i$ approaches $y(i/n)^{1/(n-1)}$ pointwise for all $i\in[n]$~\citep[Theorem C]{Kertz86}, where $y$ is the unique solution to the following ordinary differential equation
\[
	y' = y\cdot\ln(y) - y - \beta +1 \quad \text{and} \quad y(0) = 1,
\]
where $\beta \approx 1.3414 \approx 1/0.745$. This solution turns out to be decreasing and convex. Then, conditional on reaching random variable $X_i$, it chooses a quantile $q_i\in[1-x_{i-1},1-x_i]$ according to the probability density function $$f_i(q)=\frac{(n-1) (1-q)^{n-2}}{\alpha_i}\quad\text{where}\quad\alpha_i=\int_{1-x_{i-1}}^{1-x_i} (n-1) (1-r)^{n-2}\;\mathrm{d}r\,,$$
and sets $F^{-1}(1-q)$ as threshold for accepting $X_i$.

Now let $\tau$ be the stopping time implied by the algorithm, and for $q \in [0,1]$ define $R(q)=\EE[X\mid X\geq F^{-1}(1-q)]$ to be the expected value of random variable $X$ conditioned on $X$ exceeding threshold $F^{-1}(1-q)$. It can then be shown that 
\begin{align}
\E{X_{\tau}} &= \sum_{i=1}^n \Pr{\tau \geq i} \int_{1-x_{i-1}}^{1-x_i} f_i(q) R(q) q \;\mathrm{d}q \notag\\
&= \sum_{i=1}^n \rho \cdot \alpha_i \int_{1-x_{i-1}}^{1-x_i} f_i(q) R(q) q \;\mathrm{d}q = \frac{\rho}{n} \cdot  \E{\max\{X_1,\ldots,X_n\}},\label{eq:jose-ec17}
\end{align}
where the $i$th term in the sum can be viewed as the contribution of $X_i$ to the expectation $\E{X_{\tau}}$, and $\rho/n\geq \beta^{-1}$.

To simplify the presentation of our result, we first show that setting deterministic thresholds is also sufficient to achieve a guarantee of $\beta^{-1}$. In particular, for all $i\in[n]$ define $\bar{q}_i$ such that 
\[
	\bar{q}_i = \int_{1-x_{i-1}}^{1-x_i} f_i(q) q \;\mathrm{d}q\,.
\]

Let $\bar{\tau}$ be the stopping time of the algorithm that sets deterministic thresholds $F^{-1}(1-\bar{q}_i)$. The following lemma shows that we can also consider this algorithm.

\begin{lemma}\label{lem:deterministic}
	Let $X_1, \dots, X_n$ be i.i.d.\@ random variables drawn from a known distribution $F$. Then for all $i\in[n]$, we have 
	\[
	R(\bar{q}_i) \cdot \bar{q}_i \geq \int_{1-x_{i-1}}^{1-x_i} f_i(q) R(q) q \;\mathrm{d}q
	\]
\end{lemma}
\begin{proof}
	The left-hand side is $\EE[X | X > T_d] \cdot \Pr{X > T_d}$, where $T_d$ is the deterministic threshold $F^{-1}(1-\bar{q}_i)$ that corresponds to $\bar{q}_i$. The right-hand side is $\EE[X | X > T_r] \cdot \Pr{X > T_r}$, where $T_r$ is the randomized threshold that arises from first drawing $q \in [1-x_{i-1},1-x_i]$ with probability $f_i(q)$ and then setting threshold $F^{-1}(1-q)$. 
	We have chosen $\bar{q}_i$ so that $\Pr{X > T_d} = \Pr{X > T_r}$, it thus suffices to show that $\EE[X | X > T_d] \geq  \EE[X | X > T_r]$. 
	
	Let us prove the stronger statement that for all $t$, $\Pr{X > t \mid X > T_d} \geq \Pr{X > t \mid X > T_r}$. Indeed, if $t > T_d$ the claimed inequality becomes 
	\[
	\frac{\Pr{X > t}}{\Pr{X > T_d}} \geq \frac{\Pr{X > \max\{t,T_r\}}}{\Pr{X > T_r}}, 
	\]
	which holds because $\Pr{X > T_d} = \Pr{X > T_r}$ and $\Pr{X > t} \geq \Pr{X > \max\{t,T_r\}}$. On the other hand, if $t \leq T_d$, then $\Pr{X > t \mid X > T_d} = 1$ and the claimed inequality applies as well.
	\end{proof}

By applying \autoref{lem:deterministic} 
to~\eqref{eq:jose-ec17}, we obtain
\begin{align}
\E{X_{\bar{\tau}}} 
&= \sum_{i=1}^n \Pr{\bar{\tau} \geq i} \cdot R(\bar{q}_i) \cdot \bar{q}_i \notag\\
&= \sum_{i=1}^n \rho \cdot \alpha_i \cdot R(\bar{q}_i) \cdot \bar{q}_i \geq \frac{\rho}{n} \cdot  \E{\max\{X_1,\ldots,X_n\}}\,,\label{eq:jose-ec17-new}
\end{align}
where we used that conditioned on reaching step $i$ both $\tau$ and $\bar{\tau}$ accept $X_i$ with the same probability and so $\Pr{\bar{\tau} \geq i} = \Pr{\tau \geq i} = \rho \cdot \alpha_i$.

The algorithm that achieves the bound claimed in \autoref{thm:quadratic} starts by skipping some random variables until the acceptance probability $\bar{q}_i$ of algorithm $\bar{\tau}$ becomes sufficiently large, say $\delta/n$ where $0 < \delta < 1$ is some constant. Such a step exists for sufficiently large $n$, because if all acceptance probabilities $\bar{q}_1,\dots,\bar{q}_n$ were at most $1/n$, $\Pr{\bar{\tau}\leq n}$ would be at most $1-1/e\approx 0.632$ in the limit for $n\rightarrow\infty$, contradicting $\Pr{\bar{\tau}\leq n}\rightarrow\beta^{-1}\approx 0.745$~\citep{CorreaFHOV17}.

So assume that $\ell$ is the first such step with $\bar{q}_\ell\geq\delta_n$.
From then on, it uses the empirical distribution function of the samples to estimate the quantiles $\bar{q}_{\ell}, \dots, \bar{q}_n$ used by the optimal algorithm that knows the distribution on the remaining random variables. The algorithm then accepts random variable $X_i$ conditional on reaching it with probability $\tilde{q}_i$, where $\tilde{q}_i$ is its estimate of $\bar{q}_i$. More formally, when the original algorithm chooses threshold $\bar{T}_i = F^{-1}(1-\bar{q}_i)$ so that $1-F(\bar{T}_i) = \bar{q}_i$ our algorithm will choose $\tilde{T}_i = \tilde{F}^{-1}(1-\bar{q}_i)$ where $\tilde{F}$ is the empirical distribution function, and $\tilde{q}_i = 1-F(\tilde{T}_i)$. Denote the stopping time of this algorithm by $\tilde{\tau}_\delta$.

The reason why we skip the first few random variables is because the initial acceptance probability of the optimal algorithm that knows the distribution is of the order of $1/n^2$, therefore with $n^2$ samples we cannot get a reliable estimate of the corresponding quantile. 

We will lower bound the performance of our algorithm $\tilde{\tau}_{\delta}$ in terms of the performance of algorithm $\bar{\tau}$ through an intermediate algorithm, whose stopping time we denote by $\bar{\tau}_\delta$, that also skips the first few random variables but then uses the \emph{actual} quantiles $\bar{q}_{\ell}, \dots, \bar{q}_n$.

\begin{lemma}\label{lem:part1}
For every $\varepsilon>0$, there exists an $n_\varepsilon\in\N$ such that the following holds for all $n\geq n_\varepsilon$. Let $X_1, \dots, X_n$ be i.i.d.\@ random variables drawn from an unknown distribution $F$. Then, for any $\delta$ such that $0 < \delta < 1/2$, 
\[
	\E{X_{\bar{\tau}_\delta}} \geq \left(1-2\delta \right) \cdot \E{X_{\bar{\tau}}}.
\]
\end{lemma}

\begin{proof}[Proof of \autoref{lem:part1}]
Note that if $\ell = 1$ then there is nothing to show as $\bar{\tau}_\delta$ and $\bar{\tau}$ are identical. Otherwise, $\ell \geq 2$, $\bar{q}_\ell \geq \delta/n$, and $\bar{q}_{\ell-1} \leq \delta/n$, which implies that $1-x_{\ell-2} < \delta/n$.

The expected value achieved by the algorithm that skips the first few random variables until the acceptance probability becomes $\delta/n$ and then uses the actual quantiles $\bar{q}_{\ell}, \dots, \bar{q}_n$ is
\[
\E{X_{\bar{\tau}_\delta}} 
= \sum_{i = \ell}^{n} \Pr{\bar{\tau}_\delta \geq i} \cdot R(\bar{q}_i) \cdot \bar{q}_i \geq \sum_{i = \ell}^{n} \Pr{\bar{\tau}  \geq i} \cdot R(\bar{q}_i) \cdot \bar{q}_i\,.
\]
while
\[
\E{X_{\bar{\tau}}} = \sum_{i = 1}^{\ell-1} \Pr{\bar{\tau}  \geq i} \cdot R(\bar{q}_i) \cdot \bar{q}_i + \sum_{i = \ell}^{n} \Pr{\bar{\tau}  \geq i} \cdot R(\bar{q}_i) \cdot \bar{q}_i\,.
\]

Now observe that
\begin{align*}
\sum_{i=1}^{\ell-1} \Pr{\bar{\tau} \geq i} \cdot R(\bar{q}_i) \cdot \bar{q}_i 
&\leq 
\left(\sum_{i=1}^{\ell-2} \Pr{\bar{\tau} \geq i}  +  \Pr{\bar{\tau} \geq \ell-1} \right) \cdot R(\bar{q}_\ell) \cdot \bar{q}_\ell\\
&\leq 
2 \cdot \sum_{i=1}^{\ell-2} \Pr{\bar{\tau} \geq i}  \cdot R(\bar{q}_\ell) \cdot \bar{q}_\ell\\
&= 2 \cdot \rho \cdot \left(\sum_{i=1}^{\ell-2} \alpha_i\right) \cdot R(\bar{q}_{\ell-1}) \cdot \bar{q}_{\ell-1}\\
&= 2 \cdot \rho \cdot \left(\int_{0}^{1-x_{\ell-2}} (n-1)(1-q)^{n-2} \;\mathrm{d}q \right) \cdot R(\bar{q}_\ell) \cdot \bar{q}_\ell\\
&= 2 \cdot \rho \cdot (1-{x_{\ell-2}}^{n-1}) \cdot R(\bar{q}_\ell) \cdot \bar{q}_\ell\\
&\leq 2 \cdot \rho \cdot (1-(1-\delta/n)^{n-1}) \cdot R(\bar{q}_\ell) \cdot \bar{q}_\ell\\
&\le 2 \cdot \rho \cdot (1-e^{-\delta}) \cdot  R(\bar{q}_\ell) \bar{q}_\ell\\
&\le 
2 \cdot \rho \cdot \delta \cdot  R(\bar{q}_\ell) \bar{q}_\ell\,,
\end{align*}
where for the first inequality we used that $R(q) q = \int_{0}^{q} F^{-1}(1-r) \;\mathrm{d}r$ is monotone, for the second inequality we used that $\Pr{\bar{\tau} \geq \ell-1} \leq \Pr{\bar{\tau} \geq \ell-2}$, and for the third inequality we used that $x_{\ell-2} \geq 1-\delta/n$.

On the other hand,
\begin{align*}
\sum_{i = \ell}^{n} \Pr{\bar{\tau} \geq i} R(\bar{q}_i) \bar{q}_i 
&\geq 
\sum_{i = \ell}^{n} \Pr{\bar{\tau} \geq i} R(\bar{q}_\ell) \bar{q}_\ell\\
&= \rho \left(\sum_{i=\ell}^{n} \alpha_i \right) R(\bar{q}_\ell) \bar{q}_\ell \\
&= \rho \cdot \int_{1-x_{\ell-1}}^1 (n-1)(1-q)^{n-2} \;\mathrm{d}q \cdot R(\bar{q}_\ell) \bar{q}_\ell\\
&= \rho \cdot \left(1-\int_{0}^{1-x_{\ell-1}} (n-1)(1-q)^{n-2}\;\mathrm{d}q\right) \cdot R(\bar{q}_\ell) \bar{q}_\ell\\
&= \rho \left(1-\sum_{i=1}^{\ell-1} \alpha_i\right) \cdot R(\bar{q}_\ell) \bar{q}_\ell \\
&= \left(\rho - \sum_{i=1}^{\ell-1} \Pr{\bar{\tau} \geq i}\right) \cdot R(\bar{q}_\ell) \bar{q}_\ell\\
&\geq \rho \cdot (1-2\delta) \cdot R(\bar{q}_\ell) \bar{q}_\ell\,,
\end{align*}
where we again used the monotonicity of $R(q)q$ for the first inequality and the upper bound for $\sum_{i=1}^{\ell-1} \Pr{\bar{\tau} \geq i}$ derived above for the second inequality.

This shows that the ratio between the two terms is at most $2\delta/(1-2\delta)$, which in turn implies that
\[
\E{X_{\bar{\tau}}} \leq \left(1+\frac{2\delta}{1-2\delta}\right) \E{X_{\bar{\tau}_\delta}},
\]
and after rearranging shows the claim.
\end{proof}

\begin{lemma}\label{lem:part2}
For every $\varepsilon'>0$ and $\delta\in(0,1/2)$, there exist $\varepsilon''>$, $\gamma > 0$, and $n_{\varepsilon'}\in\N$ such that the following holds for all $n\geq n_{\varepsilon'}$. Let $X_1, \dots, X_n$ be i.i.d.\@ random variables drawn from an unknown distribution $F$. Then with $k \geq n^2 \ln(2/\varepsilon'')/ (2\gamma^2)$ samples it holds that
\[
\E{X_{\tilde{\tau}_\delta}} \geq (1-\varepsilon') \cdot \E{X_{\bar{\tau}_\delta}}\,.
\]
\end{lemma}

To prove this lemma, we make use of the following auxiliary lemma, which can be proven using the Dvoretzky--Kiefer--Wolfowitz inequality~\citep{DKW}.

\begin{lemma}\label{lem:aux}
Fix $\varepsilon'' > 0$ and $\gamma > 0$. Then, with $k\geq n^2 \ln(2/\varepsilon'')/ (2\gamma^2)$ samples,
\[
	\Pr{\max_i |\bar{q}_i-\tilde{q}_i| > \frac{\gamma}{n}} < \varepsilon''\,.
\]
\end{lemma}
\begin{proof}
Let $F$ denote the true underlying distribution, and let $\tilde{F}$ denote the empirical cumulative density function from $k$ samples. With $k \geq n^2 \ln(2/\varepsilon')/ (2\gamma^2)$ samples, the Dvoretzky--Kiefer--Wolfowitz inequality~\citep{DKW} states that
\[
\Pr{\sup_x |\tilde{F}(x) - F(x)| > \frac{\gamma}{n}} \leq 2\cdot e^{-2 k (\gamma/n)^2} \leq \varepsilon''\,.
\]

So with probability at least $1-\varepsilon''$ we have that for all pairs $\bar{q}_i$ and $\tilde{q}_i = 1-F(\tilde{T}_i)$,
\begin{align*}
\tilde{q}_i 
&= 1-F(\tilde{T}_i) 
 \le 1- \tilde{F}(\tilde{T}_i) + \frac{\gamma}{n}
 =  1 - (1-\bar{q}_i) + \frac{\gamma}{n}
 = \bar{q}_i + \frac{\gamma}{n}, \quad\text{and}\\[4pt]
\tilde{q}_i 
&= 1-F(\tilde{T}_i) 
\ge 1- \tilde{F}(\tilde{T}_i) - \frac{\gamma}{n}
 =  1 - (1-\bar{q}_i) - \frac{\gamma}{n}
  = \bar{q}_i - \frac{\gamma}{n}\,,
\end{align*}
as claimed.
\end{proof}

\begin{proof}[Proof of \autoref{lem:part2}]
	Given $\varepsilon' > 0$ and $\delta \in (0,1/2)$ choose $\varepsilon'', \varepsilon''', \gamma > 0$ such that $\gamma \leq \delta$, $\gamma \leq (1-\beta^{-1})\cdot \varepsilon''' \approx 0.255 \cdot \varepsilon'''$, and $(1-\varepsilon'')\cdot(1-\varepsilon''')\cdot(1-\gamma/\delta) \geq 1-\varepsilon'$.
	
	Towards relating $\EE[X_{\tilde{\tau}_\delta}]$ and $\EE[X_{\bar{\tau}_\delta}]$, denote by $\omega$ the event that $\max_i |\bar{q}_i-\tilde{q}_i| \leq \gamma/n$.
	Note that we can lower-bound the expected value obtained by our algorithm by only considering the case that $\omega$ occurs and then summing over all steps:
	\begin{align}
		\EE[X_{\tilde{\tau}_\delta}]\geq\Pr{\omega}\cdot\sum_{i=\ell+1}^n\Pr{\tilde{\tau}_\delta \geq i\mid\omega} &\cdot \Pr{X_i\geq \tilde{F}^{-1}(1-\bar{q}_i)\mid\omega \wedge \tilde{\tau}_\delta \geq i} \notag\\ 
		&\cdot\EE[X_i \mid \omega \wedge \tilde{\tau}_\delta \geq i \wedge X_i\geq \tilde{F}^{-1}(1-\bar{q}_i)]. \label{eq:empirical1}
	\end{align}	
	
	Fix some $i\in[n]$ with $i\geq \ell+1$.
	We first bound $\Pr{\tilde{\tau}_\delta \geq i\mid\omega}$ with respect to $\Pr{\bar{\tau}_\delta \geq i\mid\omega}=\Pr{\bar{\tau}_\delta \geq i}$.
	First note that, for any $j\leq\ell$, it holds that $\Pr{\bar{\tau}_\delta=j} = \Pr{\tilde{\tau}_\delta = j} = 0$.
	Furthermore note that, for all $\ell+1\le  j<i$ and conditioned on $\omega$, we have $|\bar{q}_j-\tilde{q}_j|\leq \gamma/n$.
	Thus, conditioned on $\omega$, the probability that precisely one of $\bar{\tau}_\delta,\tilde{\tau}_\delta$ is $j$ is bounded by $\gamma/n$.  
       By the union bound, we therefore have
	\begin{align*}
		\Pr{\tilde{\tau}_\delta \geq i\mid\omega}\geq\Pr{\bar{\tau}_\delta\geq i}-(i-\ell-1)\cdot\gamma/n\geq\Pr{\bar{\tau}_\delta\geq i}-i\cdot\gamma/n.
	\end{align*}
		
	Now note that $\Pr{\bar{\tau}_\delta \geq i} \geq \Pr{\bar{\tau}_\delta > n}$, \ie the probability that $\bar{\tau}_\delta$ stops at step $i$ or later is lower bounded by the probability that $\bar{\tau}_\delta$ does not stop at all. Moreover, $\Pr{\bar{\tau}_\delta > n} \geq \Pr{\tau > n}$, \ie the probability that $\bar{\tau}_\delta$ does not stop at all is at least the probability that $\tau$ does not stop at all. Since $\Pr{\tau > n} = 1-\beta^{-1} \approx 0.255$, if we choose $\gamma$ such that $\gamma \leq (1-\beta^{-1}) \cdot \varepsilon'''$, then
	\begin{align}\label{eq:empirical2}
		\Pr{\tilde{\tau}_\delta \geq i\mid\omega}\geq  (1-\varepsilon''') \cdot \Pr{\bar{\tau}_\delta \geq i}\,.
	\end{align}

	Next, we bound $\Pr{X_i\geq \tilde{F}^{-1}(1-\bar{q}_i)\mid\omega \wedge \tilde{\tau}_\delta \geq i}\cdot\EE[X_i \mid \omega \wedge \tilde{\tau}_\delta \geq i \wedge X_i\geq \tilde{F}^{-1}(1-\bar{q}_i)]$, that is, the value extracted from $X_i$ conditioned on $\omega$ and arriving in step $i$. We obtain
	\begin{align}
		&\Pr{X_i\geq \tilde{F}^{-1}(1-\bar{q}_i)\;\big\vert\;\omega \wedge \tilde{\tau}_\delta \geq i}\cdot\EE[X_i \mid \omega \wedge \tilde{\tau}_\delta \geq i \wedge X_i\geq \tilde{F}^{-1}(1-\bar{q}_i)]\nonumber	\\
		=&\,\EE\left[\int_{\tilde{q}_i}^1 F^{-1}(q)\;\mathrm{d}q\;\Big\vert\;\omega\wedge \tilde{\tau}_\delta \geq i\right]\nonumber\\
		=&\,\EE\left[\int_{\bar{q}_i}^1 F^{-1}(q)\;\mathrm{d}q-\int_{\bar{q}_i}^{\tilde{q}_i} F^{-1}(q)\;\mathrm{d}q\;\Big\vert\;\omega\wedge \tilde{\tau}_\delta \geq i\right]\nonumber\\
		\geq&\,\EE\left[\left(1-\frac{\gamma}{\delta}\right)\cdot\int_{\bar{q}_i}^1 F^{-1}(q)\;\mathrm{d}q\;\Big\vert\;\omega\wedge \tilde{\tau}_\delta \geq i\right]\nonumber\\
		=&\,\left(1-\frac{\gamma}{\delta}\right)\cdot\int_{\bar{q}_i}^1 F^{-1}(q)\;\mathrm{d}q\nonumber\\
		=&\,\left(1-\frac{\gamma}{\delta}\right)\cdot\Pr{X_i\geq F^{-1}(1-\bar{q}_i)\mid\bar{\tau}_\delta\geq i}\cdot\EE[X_i \mid \bar{\tau}_\delta \geq i \wedge X_i\geq F^{-1}(1-\bar{q}_i)], \label{eq:empirical3}
	\end{align}
	where in the third-to-last step we used that $F^{-1}(q)$ is monotonically increasing in $q$. Note that the inequality holds independently of how $\bar{q}_i$ and $\tilde{q}_i$ are ordered.
	
	We now substitute the two bounds given in~\eqref{eq:empirical2} and~\eqref{eq:empirical3} into~\eqref{eq:empirical1} and apply \autoref{lem:aux} to obtain
	\[
		\EE[X_{\tilde{\tau}_\delta}]
		\geq \Pr{\omega} \cdot \left(1-\varepsilon'''\right) \cdot (1-\frac{\gamma}{\delta}) \cdot \EE[X_{\bar{\tau}_\delta}]
		\geq (1-\varepsilon'') \cdot (1-\varepsilon''') \cdot (1-\frac{\gamma}{\delta}) \cdot \EE[X_{\bar{\tau}_\delta}]
		\geq (1-\varepsilon') \cdot \EE[X_{\bar{\tau}_\delta}]\,,
	\]
as claimed.
\end{proof}

We are now ready to prove the theorem.

\begin{proof}[Proof of \autoref{thm:quadratic}]
First choose $\varepsilon'>0$ and $\delta\in(0,1/2)$ such that $(1-\varepsilon')(1-2\delta)\geq 1-\varepsilon$. Then, by combining \autoref{lem:part1} with \autoref{lem:part2}, we obtain that there exist $\varepsilon'', \gamma > 0$ such that for sufficiently large $n$, with $k \geq n^2 \ln(2/\varepsilon'')/ (2\gamma^2) =O(n^2)$ samples,
\begin{align*}
	\EE[X_{\tilde{\tau}_\delta}]&\geq(1-\varepsilon')(1-\delta)\cdot\EE[X_{\bar{\tau}}]\\
	&\geq (1-\varepsilon)\cdot\EE[X_{\bar{\tau}}]\\
	&\geq (\beta^{-1}-\varepsilon)\cdot\EE[\max\{X_1,\dots,X_n\}],
\end{align*}
where the last step follows from~\eqref{eq:jose-ec17-new}.
\end{proof}

\begin{appendices}

\section{Proof of \autoref{thm:secretary}}
\label{app:secretary}

Let $\tau$ be the stopping time corresponding to the optimal stopping rule for the secretary problem, which rejects a certain fraction of the random variables and uses their maximum as a threshold for the remaining ones. Since $X_1,X_2,\dots,X_n$ are drawn independently from the same distribution, we can assume that their realizations are obtained by independently drawing~$n$ values from the distribution and then ordering them according to a random permutation~$\pi$. Denoting the density of the distribution from which $X_1,\dots,X_n$ are drawn by~$f$,
	\begin{align*}
		\mathbb{E}[X_\tau] & = \int_0^\infty\dots\int_0^\infty\prod_{i=1}^n f(v_i)\cdot\EE_{\pi}[v_{\pi(\tau)}] \;\mathrm{d}v_1\cdots\mathrm{d}v_n \\
		& \geq \int_0^\infty\dots\int_0^\infty\prod_{i=1}^n f(v_i)\cdot\PP_{\pi}[v_{\pi(\tau)}=\max\{v_1,\dots,v_n\}] \cdot \max\{v_1,\dots,v_n\} \;\mathrm{d}v_1\cdots\mathrm{d}v_n \\
		&\geq \frac{1}{e} \cdot\int_0^\infty\dots\int_0^\infty\prod_{i=1}^n f(v_i)\cdot\max\{v_1,\dots,v_n\}  \;\mathrm{d}v_1\cdots\mathrm{d}v_n \\
		&= \frac{1}{e} \cdot \E{\max\{X_1,\dots,X_n\}},
	\end{align*}
	where the second inequality holds because the values $v_1,\dots,v_n$ have been randomly ordered and $\tau$ is thus guaranteed to select $\max\{v_1,\dots,v_n\}$ with probability at least $1/e$ for any realization~\citep{Ferg89a}. This proves the claim.

\section{A $\mathbf{1/2}$-approximation with $n-1$ samples}
\label{app:one-half}

We formalize the discussion in \autoref{sec:warm-up}. We show that if the stopping rule has access to $n-1$ samples, then we can simply take the maximum of these samples as a single, non-adaptive threshold for all random variables to obtain a factor $1/2$-approximation. 

\begin{theorem}\label{thm:single-threshold}
Let $X_1,X_2,\dots,X_n$ be i.i.d.\@ draws from an unknown distribution $F$. Then there exists a $(n-1,n)$-stopping-rule with stopping time $\tau$ such that
\[
	\E{X_\tau} \geq \frac{1}{2} \cdot  \E{\max\{X_1, \dots, X_n\}}.
\]
\end{theorem}

To prove \autoref{thm:single-threshold} we will analyze a slight variation of the algorithm described above, \autoref{alg:single-threshold}, which only uses the maximum of $n-1$ samples as a threshold for the first $n-1$ random variables and stops on the $n$th random variable with certainty. The advantage of this is that it becomes even clearer when and why our analysis is tight. 

\begin{algorithm}[t]
\SetInd{0.5em}{1em}
\DontPrintSemicolon
\KwData{Sequence of i.i.d.~random variables $X_1, \dots, X_n$ sampled from an unknown distribution $F$, sample access to $F$}
\KwResult{Stopping time $\tau$}
  $\tau \longleftarrow n$\;
  $S_1, \dots, S_{n-1} \longleftarrow \text{$n-1$ samples from $F$}$\;
  \For{$t = 1, \dots, n-1$}{
    \If{$X_t \geq \max\{S_1, \dots, S_{n-1}\}$}{
  	$\tau \longleftarrow t$\;
	Break\;
    }
    }
    \Return $\tau$
\caption{Single threshold algorithm}
\label{alg:single-threshold}
\end{algorithm}

\begin{proof}[Proof of \autoref{thm:single-threshold}]
The expected value achieved by \autoref{alg:single-threshold} is the sum over all time steps $i = 1, \dots, n$ of the product of the probability of stopping at this time step and the expected value of the random variable conditioned on being above the threshold
\begin{align}
	\E{X_\tau} 
	&= \sum_{i=1}^{n-1} \biggl( \E{X_i \mid \tau = i} \cdot \Pr{\tau = i}  \biggr) + \E{X_n} \cdot \Pr{\tau = n} \notag\\
	&\geq \sum_{i=1}^{n-1} \biggl( \E{X_i \mid \tau = i} \cdot \Pr{\tau = i}  \biggr).
	\label{eq:single-threshold-1}
\end{align}

We stop at time step $i$ if the maximum among the $n-1$ samples and the first $i$ random variables happens to be the $i$th random variable, and if, conditioned on this, the second maximum is among the $n-1$ samples and not the other $i-1$ random variables. Hence,
\begin{align*}
	\Pr{\tau = i} &= \frac{1}{n-1+i} \cdot \frac{n-1}{n-2+i}
\end{align*}

Summing this over all $i$ from $1$ to $n-1$ shows that the probability of stopping at one of the first $n-1$ random variables is precisely
\begin{align}
	\sum_{i=1}^{n-1} \Pr{\tau = i} 
	&= \sum_{i=1}^{n-1} \frac{1}{n-1+i} \cdot \frac{n-1}{n-2+i} = \frac{1}{2}.
	\label{eq:single-threshold-2}
\end{align}

We conclude the proof by showing that for all $i = 1, \dots, n-1$ the conditional expectation $\E{X_i \mid \tau = i}$ is at least $\E{\max\{X_1, \dots, X_n\}}$.
Let $T = \max\{S_1, \dots, S_n\}$. The algorithm stops at time step $i$ if $X_i \geq T > \max\{X_1, \dots, X_{i-1}\}$. 
So under this event $X_i$ is the maximum of $n-1+i$ random variables. And so
\begin{align}
	\E{X_i \mid \tau = i} 
	&= \E{\text{max of $n-1+i$ i.i.d.~RVs}} \geq \E{\max\{X_1, \dots, X_n\}}.
	\label{eq:single-threshold-3}
\end{align}

Substituting~\eqref{eq:single-threshold-2} and~\eqref{eq:single-threshold-3} into~\eqref{eq:single-threshold-1} completes the proof.
\end{proof}

As we have argued in the proof of \autoref{thm:single-threshold}, the probability that \autoref{alg:single-threshold} stops on one of the first $n-1$ variables is precisely $1/2$. The two potentially lossy steps are that we dropped the contribution from the final random variable, and that we lower bounded the contribution from each of the first $n-1$ random variables by $\E{\max\{X_1, \dots, X_n\}}$.

It turns out that both of the potentially lossy steps are in fact lossless in the limit as $n \rightarrow \infty$ if $F$ is the exponential distribution.

\begin{proposition}
Let $X_1, \dots, X_n$ be drawn independently from the distribution with $F(x)=1-e^{-x}$. Then for the stopping time $\tau$ determined by \autoref{alg:single-threshold},
\begin{align*}
	\lim_{n \rightarrow \infty} \frac{\E{X_\tau}}{\E{\max\{X_1, \dots, X_n\}}} = \frac{1}{2}.
\end{align*}
\end{proposition}
\begin{proof}
It is a well-known fact that the maximum of $n$ independent, exponentially distributed random variables $X_1, \dots, X_n$ is equal to the $n$th harmonic number, \ie that
\begin{align*}
	\E{\max\{X_1, \dots, X_n\}} = H_n .
\end{align*}

As we have argued in the proof of \autoref{thm:single-threshold}, the expected value obtained by \autoref{alg:single-threshold} can be written as
\begin{align*}
	\E{X_\tau} 
	&= \sum_{i=1}^{n-1} \biggl( H_{n-1+i} \cdot \frac{1}{n-1+i} \cdot \frac{n-1}{n-2+i}  \biggr) + \frac{1}{2}.
\end{align*}
Tedious calculations allow to express the expected value via the digamma function $\psi^{(0)}$ and the Euler-Mascheroni constant $\gamma$ as
\begin{align*}
	\E{X_\tau} &= \psi^{(0)}(n) - \frac{1}{2} H_{2n-2} + \gamma + 1 ,
\end{align*}
which can be used to show that
\begin{align*}
\lim_{n \rightarrow \infty} \frac{\E{X_\tau}}{\E{\max\{X_1, \dots, X_n\}}} = \lim_{n \rightarrow \infty} \frac{\psi^{(0)}(n) - \frac{1}{2} H_{2n-2} + \gamma + 1}{H_n} = \frac{1}{2} .
\end{align*}
This proves the claim.
\end{proof}

\end{appendices}

\section*{Acknowledgments}
Discussions with F\'abio Botler and valuable feedback from the anonymous referees are gratefully acknowledged. An extended abstract announcing the main results appeared in the Proceedings of the 20th ACM Conference on Economics and Computation. The work was funded in part by an Amazon Research Award, by CONICYT grants AFB-170001 and FONDECYT-190043, by the DAAD within the PRIME program using funds of BMBF and the EU Marie Curie Actions, by the European Research Council under Grant Agreement No.~691672, and by EPSRC grant~EP/T015187/1.  Part of the work was done while the second author was at Google Research in Z\"urich, Switzerland, and while the fourth author was at \'Ecole Normale Sup\'erieure in Paris, France, and Technische Universit\"at M\"unchen, Germany.

\bibliographystyle{abbrvnat}
\bibliography{abb,bibliography}

\end{document}